\documentclass[a4paper, 11pt]{article}
\usepackage[usenames]{xcolor}
\usepackage{xifthen}

\def\fontsettingup{2} 

\usepackage{amsmath, amsthm, amssymb}
\usepackage[T1]{fontenc}
\ifthenelse{\fontsettingup = 1}{ \usepackage{eulerpx, eucal, tgpagella}   }{}
\ifthenelse{\fontsettingup = 2}{ \usepackage{mathpazo, eufrak, tgpagella} }{}

\let\OLDthebibliography\thebibliography
\renewcommand\thebibliography[1]{
  \OLDthebibliography{#1}
  \setlength{\parskip}{0pt}
  \setlength{\itemsep}{5.5pt}
}

\usepackage{enumitem}
\setlist{itemsep=6pt, topsep = 6pt, parsep = \parskip, partopsep = 6pt}

\usepackage{hyperref}
\hypersetup{
  colorlinks=true,
  linkcolor=blue,
  filecolor=magenta,
  urlcolor=cyan,
  citecolor=blue
}

\usepackage{algorithm2e}
\usepackage{tikz}
\usepackage{cleveref}

\usepackage[a4paper]{geometry}
 \newgeometry{
 textheight=9in,
 textwidth=6.5in,
 top=1in,
 headheight=14pt,
 headsep=25pt,
 footskip=30pt
}

\newtheorem{theorem}{Theorem}

\newtheorem*{claim*}{Claim}

\newtheorem{lemma}[theorem]{Lemma}
\newtheorem{proposition}[theorem]{Proposition}

\theoremstyle{definition}

\newtheorem{definition}[theorem]{Definition}
\newtheorem{remark}[theorem]{Remark}

\newtheorem*{remark*}{Remark}

\newtheorem{conjecture}[theorem]{Conjecture}

\ifthenelse{\fontsettingup = 1}{
  \def\*#1{\mathbf{#1}} 
  \def\+#1{\mathcal{#1}} 
  \def\-#1{\mathrm{#1}} 
  \def\^#1{\mathbb{#1}} 
  \def\!#1{\mathfrak{#1}} 
}{}

\ifthenelse{\fontsettingup = 2}{
  \def\*#1{\boldsymbol{#1}} 
  \def\+#1{\mathcal{#1}} 
  \def\-#1{\mathrm{#1}} 
  \def\^#1{\mathbb{#1}} 
  \def\!#1{\mathfrak{#1}} 
}{}

\def\oPr{\mathbf{Pr}}
\renewcommand{\Pr}[2][]{ \ifthenelse{\isempty{#1}}
  {\oPr\left[#2\right]}
  {\oPr_{#1}\left[#2\right]} } 

\def\oE{\mathbb{E}}
\newcommand{\E}[2][]{ \ifthenelse{\isempty{#1}}
  {\oE\left[#2\right]}
  {\oE_{#1}\left[#2\right]} }

\def\oVar{\mathbf{Var}}
\newcommand{\Var}[2][]{ \ifthenelse{\isempty{#1}}
  {\oVar\left[#2\right]}
  {\oVar_{#1}\left[#2\right]} }

\def\oEnt{\mathbf{Ent}}
\newcommand{\Ent}[2][]{ \ifthenelse{\isempty{#1}}
  {\oEnt\left[#2\right]}
  {\oEnt_{#1}\left[#2\right]} }

\newcommand{\PhiEnt}[2][]{ \ifthenelse{\isempty{#1}}
  {\oEnt^\phi\left[#2\right]}
  {\oEnt^\phi_{#1}\left[#2\right]} }

\newcommand{\DKL}[2]{\+D_{\-{KL}}\left(#1 \parallel #2\right)}
\newcommand{\DTV}[2]{\+D_{\mathrm{TV}}\left({#1} \parallel {#2}\right)}
\newcommand{\dist}{\mathrm{dist}}
\newcommand{\e}{\mathrm{e}}
\renewcommand{\epsilon}{\varepsilon}
\renewcommand{\emptyset}{\varnothing}

\newcommand{\set}[1]{\left\{#1\right\}}
\newcommand{\tuple}[1]{\left(#1\right)} \newcommand{\eps}{\varepsilon}

\newcommand{\tp}{\tuple}
\newcommand{\ol}{\overline}
\newcommand{\abs}[1]{\left\vert#1\right\vert}
\newcommand{\ctp}[1]{\left\lceil#1\right\rceil}

\title{Rapid Mixing via Coupling Independence \\ for Spin Systems with Unbounded Degree}

\newboolean{doubleblind}
\setboolean{doubleblind}{false}

\ifdoubleblind
\author{Author(s)}
\else

\author{
Xiaoyu Chen\thanks{State Key Laboratory for Novel Software Technology, New Cornerstone Science Laboratory, Nanjing
University, China. E-mail: chenxiaoyu233@smail.nju.edu.cn}
\and
Weiming Feng\thanks{Institute for Theoretical Studies, ETH Z\"{u}rich, Switzerland. Email: weiming.feng@eth-its.ethz.ch. Research supported by Dr. Max R\"ossler, the Walter Haefner Foundation and the ETH Z\"urich Foundation.}
}

\fi

\date{}

\begin{document}

\maketitle

\begin{abstract}
    We develop a new framework to prove the mixing or relaxation time for the Glauber dynamics on spin systems with unbounded degree.
    It works for general spin systems including both $2$-spin and multi-spin systems.
    As applications for this approach:
    \begin{itemize}
    \item We prove the optimal $O(n)$ relaxation time for the Glauber dynamics of random $q$-list-coloring on an $n$-vertices triangle-tree graph with maximum degree $\Delta$ such that $q/\Delta > \alpha^\star$, where $\alpha^\star \approx 1.763$ is the unique positive solution of the equation $\alpha = \exp(1/\alpha)$. 
        This improves the $n^{1+o(1)}$ relaxation time for Glauber dynamics obtained by the previous work of Jain, Pham, and Vuong (2022).
        Besides, our framework can also give a near-linear time sampling algorithm under the same condition.
        \item We prove the optimal $O(n)$ relaxation time and near-optimal $\widetilde{O}(n)$ mixing time for the Glauber dynamics on hardcore models with parameter $\lambda$ in \emph{balanced} bipartite graphs such that $\lambda < \lambda_c(\Delta_L)$ for the max degree $\Delta_L$ in left part and the max degree $\Delta_R$ of right part satisfies $\Delta_R = O(\Delta_L)$. This improves the previous result by Chen, Liu, and Yin (2023).
    \end{itemize}
    At the heart of our proof is the notion of \emph{coupling independence} which allows us to consider multiple vertices as a huge single vertex with exponentially large domain and do a ``coarse-grained'' local-to-global argument on spin systems.
    The technique works for general (multi) spin systems and helps us obtain some new comparison results for Glauber dynamics.
\end{abstract}


\section{Introduction}\label{sec:intro}
The \emph{spin system} is a fundamental probabilistic graphical model. 
It is defined on a graph $G=(V,E)$, where every vertex is a random variable and every edge models the local interactions.
Each variable takes a value from a discrete domain $[q] =\{1,2,\ldots,q\}$.
Each vertex has a vector $b \in \mathbb{R}_{\geq 0}^q$ called the \emph{external field} and each edge has a symmetric matrix $A \in \mathbb{R}_{\geq 0}^{q \times q}$ called the \emph{interaction matrix}.
The spin system defines a Gibbs distribution over $[q]^V$ such that for any configuration $\sigma \in [q]^V$,
\begin{align*}
    \mu(\sigma) \propto \prod_{v \in V} b(\sigma_v) \prod_{e = \{u,v\} \in E} A(\sigma_u,\sigma_v).
\end{align*}
The spin system covers many important distributions including the uniform distribution of graph colorings, the Ising model, the hardcore gas model in Physics, and a broad class of undirected graphical models in machine learning~\cite{mezard2009information}. 

Sampling from the Gibbs distribution is a central algorithmic task for spin systems.
The \emph{Glauber dynamics} is a fundamental Markov chain Monte Carlo (MCMC) method for sampling from high-dimensional distributions. 
Given a distribution $\mu$ over $[q]^V$, it starts from an arbitrary $X \in \Omega(\mu)$, where $\Omega(\mu) \subseteq [q]^V$ is the support of $\mu$. In each step, it updates the current state $X$ as follows:
\begin{itemize}
    \item pick a variable $v \in V$ uniformly at random;
    \item resample the value of $X_v$ from the conditional distribution $\mu_{v}(\cdot \mid X_{V \setminus v})$.
\end{itemize}
It is well-known that if the state space $\Omega(\mu)$ is connected through the moves of Glauber dynamics, then the distribution $\mu$ is the unique stationary distribution for the Glauber dynamics.

%
%
%
%

In this paper, we study the convergence rate of the Glauber dynamics.
Let $(X_t)_{t \geq 0}$ denote the random sequence generated by the Glauber dynamics.
Many notions capture the convergence rate. The most standard one is the \emph{mixing time}, which is defined by
\begin{align}\label{eq:def-mix}
 T_{\textnormal{mix}}^{\textnormal{GD}}(\mu,\eps) := \max_{X_0 \in \Omega(\mu)}  \min\left\{t > 0 \mid \DTV{X_t}{\mu} \leq \eps \right\},
\end{align}
where $\DTV{X_t}{\mu}$ denote the standard \emph{total variation distance} between $\mu$ and the distribution of $X_t$.
In words, if the Glauber dynamics starts from the worst initial state $X_0$, the mixing time is the minimum number $t$ such that the total variation distance between $X_t$ and $\mu$ is below a sufficiently small constant.  
%
Another widely used notion is the \emph{relaxation time}.
Let $P: \Omega(\mu) \times \Omega(\mu) \to [0,1]$ denote the transition matrix of the Glauber dynamics. 
A standard fact says that $P$ only has non-negative real eigenvalues $1 = \lambda_1 \geq \lambda_2 \geq \ldots \geq \lambda_{|\Omega|} \geq 0$~\cite{DGU14}.
The gap $\lambda_1 - \lambda_2 = 1 - \lambda_2$ is called the \emph{spectral gap} of  Glauber dynamics. 
The relaxation time is defined by 
\begin{align*}
    T_{\text{rel}}^{\text{GD}}(\mu) := \frac{1}{1 - \lambda_2}.
\end{align*}
It is well known that  
$T_{\textnormal{mix}}^{\textnormal{GD}}(\mu,\eps)  = O ( T_{\text{rel}}^{\text{GD}}(\mu) \log \frac{1}{\eps\mu_{\min}} )$, where $\mu_{\min} = \min_{\sigma \in \Omega(\mu)}\mu(\sigma)$.
%


Recently, in a series of works~\cite{ALOV19, alev2020improved,anari2020spectral} studied Glauber dynamics using high dimensional expanders.
An important notion called \emph{spectral independence} was developed during this process.
Anari, Liu, and Oveis Gharan~\cite{anari2020spectral} first introduced spectral independence for Boolean distributions. The follow-up works~\cite{chen2021rapidcolor, feng2021rapid} then generalized it to non-Boolean distributions.
For example, for a Boolean distribution $\mu$ over $\set{-1,+1}^{[n]}$, the \emph{influence matrix} $\Psi \in \^R^{n\times n}_{\geq 0}$ is defined by
$\Psi(u, v) := \Pr[X\sim\mu]{X_v = + \mid X_u = +} - \Pr[X\sim \mu]{X_v = + \mid X_u = -}$.
A distribution $\mu$ is $C$-spectrally independent if the maximum eigenvalue of $\Psi$ is at most $C$.
If every conditional distribution of $\mu$ is $C$-spectrally independent, then by the local-to-global argument~\cite{alev2020improved}, both relaxation and mixing time of Glauber dynamics are bounded by $n^{O(C)}$, where $n$ is the number of variables.
Given this polynomial bound $n^{O(C)}$, many works tried to obtain an improved or even the optimal mixing/relaxation time for Glauber dynamics, especially when $\mu$ is a Gibbs distribution defined by spin systems.
Chen, Liu and Vigoda~\cite{chen2020optimal} proved that for spin systems on bounded degree graphs, the spectral independence implies both $O(n\log n)$ optimal mixing time and $O(n)$ optimal relaxation time.




The next question is how to deal with spin systems on \emph{unbounded degree} graphs. 
Many works \cite{jain2021spectral,chen2021rapid, AnariJKPV22, ChenE22, Chen0YZ22} focused on this question. Significant progress was made, especially for 2-spin systems ($q = 2$). 
%
\cite{jain2021spectral} first studied coloring and weighted independent sets (hardcore model) in high-girth graphs and proved the near-optimal $n^{1+o(1)}$ relaxation time.
\cite{chen2021rapid} introduced a stronger variant of spectral independence called \emph{complete spectral independence}, and proved the optimal $O(n)$ relaxation time for anti-ferromagnetic 2-spin systems in the uniqueness regime.
%
To obtain the optimal mixing time, \cite{AnariJKPV22} made the first step and defined a new notion called \emph{entropic independence}.
After a line of works~\cite{AnariJKPV22,ChenE22,Chen0YZ22}, the optimal $O(n \log n)$ mixing time was established for a broad class of 2-spin systems.


Most of the previous techniques~\cite{chen2021rapid, AnariJKPV22, ChenE22, Chen0YZ22} for unbounded degree graphs are restricted to the $2$-spin systems.
We consider the following question in this paper.

\vspace{2ex}
{\centering \it
How to prove the optimal mixing/relaxation time \\
for Glauber dynamics on (multi) spin systems with unbounded degree?
\par}
\vspace{2ex}
\noindent
To the best of our knowledge, the only previous result beyond 2-spin systems is the $n^{1+o(1)}$ relaxation time for graph coloring~\cite{jain2021spectral}. 
However, ~\cite{jain2021spectral} relies on the coupling analysis for colorings in~\cite{hayes2006coupling}, which makes it difficult to be generalized to other spin systems.
%

In this work, we develop a new framework for proving mixing/relaxation time for the Glauber dynamics on general spin systems including both $2$-spin and multi-spin systems.
Our new framework is based on a stronger variant of the spectral independence known as the \emph{coupling independence},
which is already used implicitly or explicitly in many previous works~\cite{liu2021coupling, blanca2021mixing, ChenZ23, chen2024fast, chen2023strong, jerrum2024glauber}.
A spin system $\mu$ on $[q]^V$ is $C$-coupling independent if for any $v \in V$ and $a, b \in [q]$, there is a coupling $(X, Y)$ where $X \sim \mu^{v\gets a}$ and $Y \sim \mu^{v \gets b}$ such that 
\begin{align*}
    \E{d_H(X, Y)} \leq C.
\end{align*}
Here, $d_H(X,Y) = |\{v \in V \mid X_v \neq Y_v\}|$ denotes the hamming distance between $X$ and $Y$ and $\mu^{v \gets a}$ the distribution induced by $\mu$ conditional on $v$ taking the value $a$.

Given a spin system on a graph $G$ with Gibbs distribution $\mu$,
we show that if $\mu$ and all the conditional distributions induced by $\mu$ satisfy the coupling independence and the maximum degree of $G$ is greater than a large constant, then the following comparison results hold for Glauber dynamics.
\begin{itemize}
    \item \textbf{Relaxation time comparison.} The relaxation time satisfies $T_{\text{rel}}^{\text{GD}}(\mu) = O(T_{\text{rel}}^{\text{GD}}(\mu^\star) )$, where $\mu^\star$ is a conditional distribution obtained from $\mu$ by fixing the values on a subset $\Lambda \subseteq V$ of variables. The set $\Lambda$ is chosen intentionally such that the induced subgraph $G[V\setminus \Lambda]$ on other vertices has smaller maximum degree. For many spin systems, the distribution $\mu^\star$ is in an ``easy regime'' so that the mixing/relaxation time for $\mu^\star$ is easy to analyze. We can bound the relaxation time for $\mu$ via this comparison result (see \Cref{thm:main}).
    \item \textbf{Mixing time comparison.} If $\mu$ is a monotone spin system (\Cref{def:monotone-spin-system}) and the Glauber dynamics starts from a specific initial configuration, then the mixing time satisfies $T_{\textnormal{mix}}^{\textnormal{GD}}(\mu,\eps) = \widetilde{O}(T_{\textnormal{mix}}^{\textnormal{GD}}(\mu^\star,\frac{1}{4e}))$, where $\widetilde{O}$ hides a $\mathrm{polylog}(n/\epsilon)$ factor (See \Cref{thm:mixcom}).
\end{itemize}
We obtain the relaxation/mixing time bounds via the above comparison results. In the relaxation time comparison result, the constant factor in $O(\cdot)$ is independent of the degree of the graph. 
In applications, the distribution $\mu^\star$ is in an ``easy regime'', we can use some standard technique to show $T_{\text{rel}}^{\text{GD}}(\mu^\star) = O(n)$.
The comparison result gives the optimal $T_{\text{rel}}^{\text{GD}}(\mu) = O(n)$ relaxation time. 
Similarly, in the applications of monotone systems, we can obtain the near-optimal $\widetilde{O}(n)$ mixing time for general graphs. 
Our comparison results only hold for graphs with large maximum degrees.
It does not cause any issue in applications, because coupling independence implies spectral independence, and for graphs with bounded maximum degree, \cite{chen2020optimal} already established the optimal relaxation/mixing time.


Our proving techniques can also give a near-linear time (in input size) sampling algorithm (see \Cref{thm:algo}).  
Furthermore, we introduce a general technique to establish coupling independence for 2-spin systems (\Cref{thm:informal}). 
Specifically, many spectral independence results for 2-spin systems are proved by analyzing the decay of correlation in the self-avoiding walk tree~\cite{chen2020rapid,chen2023uniqueness}.
We show that all of such proofs can be translated to a proof of coupling independence.

\paragraph{Organization of the paper} In \Cref{sec:apps}, we first exhibit some concrete applications. In \Cref{sec:mainresult}, we give our technical results and an overview of proof techniques. \Cref{sec:prelim} is for preliminaries. 
In \Cref{sec:relcom}, we prove the relaxation time comparison result. In \Cref{sec:algo}, we give a near-liner time sampling algorithm. In \Cref{sec:mixcom}, we prove the mixing time comparison result for monotone systems. In \Cref{sec:CI}, we give a general technique for establishing coupling independence. \Cref{sec:bihardcore} and \Cref{sec:ListColoring} are for proofs of applications.

\subsection{Applications}\label{sec:apps}

Let $G=(V,E)$ be a graph with maximum degree $\Delta$ and $[q] =\{1,2,\ldots,q\}$ a set of colors. 
Given a set of color lists $L_v \subseteq [q], v \in V$,
a proper list-coloring $X \in [q]^V$ assigns a color $X_v \in L_v$ to each vertex $v \in V$ such that adjacent vertices receive different colors.
In a special case when $L_v = [q]$ for all $V \in V$, the list coloring becomes the standard graph $q$-coloring.
We use $\mu$ to denote the uniform distribution over all proper list-colorings in $G$.
For the list coloring, in each step, the Glauber dynamics picks a random vertex $v$ and update its color to a random available color.
There is a long line of works studying the mixing and relaxation time of Glauber dynamics e.g.~\cite{jerrum1995very,vigoda2000improved,chen2019improved}.

In the era of spectral independence, the proper list-coloring has been re-studied by a series of works~\cite{chen2021rapidcolor, feng2021rapid, chen2023strong}. 
Though the technique varies, all these works ended up establishing some coupling independence results for the proper list-coloring.
For list colorings on triangle-free graphs.
Let $\alpha^\star \approx 1.763$ denote the unique positive solution to the equation $\alpha = \exp(1/\alpha)$.
When $|L_v| > (\alpha^\star + \delta) \Delta$, the $O_\delta(1)$-coupling independence can be established by techniques in~\cite{chen2021rapidcolor, feng2021rapid}.
Our framework gives the optimal relaxation time of Glauber dynamics even if the maximum degree of $G$ is unbounded.
\begin{theorem}[Coloring: Relaxation Time]\label{main:coloring}
 Let $\delta > 0$ be a constant. 
 For any triangle-free graph $G=(V,E)$ and color lists $(L_v)_{v \in V}$, if $|L_v| \geq (\alpha^\star + \delta)\Delta$ for all $v \in V$, where $\Delta \geq 3$ is the maximum degree of $G$,
then relaxation time of Glauber dynamics is $O_\delta(n)$, where $n$ is the number of vertices in $G$.
\end{theorem}

Under the condition of~\Cref{main:coloring}, the relaxation time of the Glauber dynamics has been studied by many previous works. Combining the spectral independence technique~\cite{alev2020improved,anari2020spectral} with the correlation decay analysis~\cite{GMP05,gamarnik2013strong}, two independent works~\cite{chen2021rapidcolor, feng2021rapid} proved the polynomial relaxation time $n^{O(1/\delta)}$ of Glauber dynamics.
For graphs with bounded maximum degree $\Delta = O(1)$, Chen, Liu and Vigoda~\cite{chen2020optimal} established the $O_{\Delta,\delta}(n)$ relaxation time, where $O_{\Delta,\delta}(\cdot)$ hides a constant factor like $\Delta^{O(\Delta^2/\delta)}$.
For general graphs with possibly unbounded maximum degree,  Jain, Pham and Vuong~\cite{jain2021spectral} proved the first almost linear relaxation time $O_{\delta}(n e^{(\log \log n)^2}) = O_{\delta}(n^{1+o(1)})$. 
Their proof combined the techniques in \cite{chen2020optimal} with the coupling analysis in~\cite{hayes2006coupling}.
Compared to previous results, \Cref{main:coloring} gives the optimal linear relaxation time for general graphs.

We prove \Cref{main:coloring} by first verifying the coupling independence condition (\Cref{def:CI}) and then applying our comparison result (\Cref{thm:main}).
\Cref{main:coloring} requires $|L_v| \geq (\alpha^\star + \delta) \Delta$ because the current best coupling independence result requires this number of colors but our comparison result does not require such a strong condition.
It is conjectured that $O_\delta(1)$-coupling independence should hold for proper list-coloring in general graphs when $|L_v| \geq (1 + \delta)\Delta + O(1)$.

\begin{conjecture}[Folklore] \label{conj:CI-list-coloring}
    Let $\delta > 0$ be a constant.
    For any graph $G = (V, E)$ with maximum degree $\Delta$ and color lists $(L_v)_{v \in V}$ such that $\abs{L_v} \geq (1 + \delta) \Delta +O(1)$ for all $v \in V$, the the uniform distribution $\mu$ over all the proper list-colorings of $G$ is $O_\delta(1)$-coupling independent.
\end{conjecture}

Our comparison framework can prove optimal relaxation time for Glauber dynamics on proper list-colorings of graphs (with potentially unbounded degree) as long as \Cref{conj:CI-list-coloring} holds.

\begin{proposition}
\label{thm:Trel-from-conj}
    If \Cref{conj:CI-list-coloring} holds with $\delta > 0$, then for any list coloring instance in \Cref{conj:CI-list-coloring}, the relaxation time of Glauber dynamics is $O_\delta(n)$.
\end{proposition}

The standard relation between relaxation time and mixing time implies that the Glauber dynamics mixes in time $O_\delta(n^2 \log q)$, which yields a sampling algorithm for the uniform distribution $\mu$ of graph colorings in time $O_\delta(\Delta n^2 \log q)$ because each step of Glauber dynamics can be simulated in time $O(\Delta)$. 
However, in terms of sampling algorithm, our technique would directly give an algorithm (not the Glauber dynamics) in time $\widetilde{O}_\delta(\Delta n)$.
Since the input graph $G$ contains $\Delta n$ edges, 
the running time is linear-near in the input size.

\begin{theorem}[Coloring: Algorithm]\label{main:coloring-algo}
 Let $\delta > 0$ be a constant. 
 There exists an algorithm such that given any $\epsilon >0$, any triangle-free graph $G=(V,E)$ and color lists $(L_v)_{v \in V}$, if $|L_v| \geq (\alpha^\star + \delta)\Delta$ for all $v \in V$, where $\Delta \geq 3$ is the maximum degree of $G$, it returns a random sample $X$ satisfying $\DTV{X}{\mu} \leq \epsilon$ in time $\Delta n (\log \frac{n}{\epsilon})^{C(\delta)}$, where $C(\delta)$ is a constant depending only on $\delta$.
\end{theorem}


The next example is the hardcore model. 
Let $G=(V,E)$ be a graph. Let $\lambda>0$ be the fugacity. The hardcore model defines a distribution $\mu$ over all independent sets $S \subseteq V$ in $G$ such that $\mu(S) \propto \lambda^{|S|}$.
Let $\Delta \geq 3$ denote the maximum degree of graph $G$.
There is a critical threshold for the tree uniqueness phase transition~\cite{kelly1985stochastic}
\begin{align*}
    \lambda_c(\Delta) := \frac{(\Delta-1)^{(\Delta-1)}}{(\Delta-2)^\Delta}.
\end{align*}
such that if $\lambda \leq \lambda_c(\Delta)$ the correction between two vertices decays in their distance; if $\lambda > \lambda_c(\Delta)$, the long-range correlation exists. A computational phase transition occurs at the same threshold. If $\lambda < \lambda_c(\Delta)$,  polynomial time sampling algorithm exists ~\cite{weitz2006counting}; if $\lambda > \lambda_c(\Delta)$, the sampling problem is hard unless \textbf{NP} = \textbf{RP}~\cite{sly2010computational}.
The mixing and relaxation time of the Glauber dynamics for hardcore model were also extensively studied~\cite{luby1999fast,hayes2006coupling, efthymiou2016convergence}. 
Recent works analyzed Glauber dynamics via spectral independence~\cite{anari2020spectral}.
The optimal $O_\delta(n \log n)$ mixing time and the optimal $O_\delta(n)$ relaxation time  were established when $\lambda \leq (1-\delta) \lambda_c(\Delta)$ for general graphs~\cite{chen2020optimal,ChenE22,Chen0YZ22}.


However, for the hardcore model on bipartite graphs, the picture is not very clear.
Consider the hardcore model in a bipartite graph $G=(V = V_L \uplus V_R, E)$.
Let $\Delta_L$ denote the maximum degree in the left part.
Assume $3 \leq \Delta_L$.
It is recently known that the uniqueness threshold for the hardcore model on the bipartite graph can be refined to $\lambda_c(\Delta_L) \geq \lambda_c(\Delta)$ where $\Delta \geq \Delta_L$ is the maximum degree of the bipartite graph ~\cite{liu2015fptas, chen2023uniqueness}.
The Glauber dynamics is also proved to have polynomial mixing time when $\lambda < \lambda_c(\Delta_L)$~\cite{chen2023uniqueness}.

On the other side, when $\lambda > \lambda_c(\Delta_L)$, the lower bound in~\cite{sly2010computational} does not hold for bipartite graphs and the problem is  \#BIS-hard~\cite{cai2016bis}, where \#BIS is the problem of counting the independent sets in bipartite graphs.
A line of works (e.g.~\cite{ jenssen2020algorithm, cannon2020counting, 
liao2019counting,
chen2022sampling,jenssen2023approximate}) studied various sampling algorithms in the low-temperature (large $\lambda$) regime.

%

%
Within the critical threshold $\lambda < \lambda_c(\Delta_L)$, we consider ``balanced'' bipartite graphs.
Let $\Delta_R$ be the maximum degree in $V_R$.
We say a bipartite graph is $\theta$-balanced if $\Delta_R \leq \theta\Delta_L$.

%
%
%
\begin{theorem}[Bipartite Hardcore: Relaxation Time]\label{thm:bipart}
    Let $\delta \in (0,1)$ and $\theta > 1$ be two constants. For any hardcore model on a $\theta$-balanced bipartite graph $G$ with fugacity $\lambda$, if $\lambda \leq (1-\delta)\lambda_c(\Delta_L)$, then the relaxation time of Glauber dynamics is $O_{\delta,\theta}(n)$, where $n$ is the number of vertices in $G$. 
\end{theorem}

For the mixing time, again, the standard relation gives $O_{\delta,\theta}(n^2)$ mixing time of the Glauber dynamics. 
However, since the bipartite graph hardcore is a monotone system, our technique also implies the $\widetilde{O}_{\delta,\theta}(n)$ \emph{mixing time} of Glauber dynamics starting from the independent set containing all vertices in the left part: $X_0 = V_L$.  Formally, for any $S \in \Omega(\mu)$, 
\begin{align*}
    T_{\textnormal{mix}}^{\textnormal{GD}}(\mu,\eps \mid S) = \min\left\{t > 0 \mid \DTV{X_t}{\mu} \leq \eps \land X_0 = S \right\}.
\end{align*}
\begin{theorem}[Bipartite Hardcore: Mixing Time]\label{thm:bipartmix}
    Let $\delta \in (0,1)$ and $\theta > 1$ be two constants. For any hardcore model on a $\theta$-balanced bipartite graph $G$ with fugacity $\lambda$, if $\lambda \leq (1-\delta)\lambda_c(\Delta_L)$, then the mixing time of Glauber dynamic starting from $V_L$ satisfies $T_{\textnormal{mix}}^{\textnormal{GD}}(\mu,\eps \mid V_L) = n (\log \frac{n}{\eps})^{{C(\delta,\theta)}}$, where $C(\delta,\theta)$ is a constant depending only on $\delta$ and $\theta$. 
\end{theorem}

%
The previous work~\cite{chen2023uniqueness} established the $(\frac{\Delta_L \log n}{\lambda})^{O(1/\delta)} n^2$ relaxation time for the Glauber dynamics, which, by the standard relation, implies the $(\frac{\Delta_L \log n}{\lambda})^{O(1/\delta)} n^3 \log \frac{1+\lambda}{\lambda}$ mixing time. 
The previous result holds for general bipartite graphs as long as $\lambda \leq (1-\delta)\lambda_c(\Delta_L)$. 
%
For balanced bipartite graphs, we obtained the optimal relaxation time $O_{\delta,\theta}(n)$ and the near-optimal mixing time $\widetilde{O}_{\delta,\theta}(n)$, which significantly improved the dependency to $n$ and $\Delta_L$ compared to the previous result. 
For example, in the critical case when $\lambda_c = (1-\delta)\lambda_c(\Delta_L) = \Theta(1/\Delta_L)$,  previous result gives $\Delta_L^{O(1/\delta)} n^2 \cdot \mathrm{polylog}(n)$ relaxation time and $\Delta_L^{O(1/\delta)} n^3 \cdot \mathrm{polylog}(n)$ mixing time but our result gives $O(n)$ relaxation time and $n \cdot  \mathrm{polylog}(n)$ mixing time. 
However, our result works only on balanced bipartite graphs. 
The result in~\cite{chen2023uniqueness} is still state-of-the-art for general bipartite graphs.  

Finally, we point out that our technique could also recover many previous $O(n)$ relaxation time results for anti-ferromagnetic 2-spin systems
in~\cite{chen2021rapid}.
See \Cref{remark-hard} for one example.

\section{Technical Results and Proof Overview}\label{sec:mainresult}

\subsection{Coupling Independence}
In this section, we give our general results for spin systems.
Let $G=(V,E)$ be a graph. Let $[q] =\{1,2,\ldots,q\}$ be a set of $q \geq 2$ spins. 
For each vertex $v \in V$, let vector $b_v \in \mathbb{R}_{\geq 0}^q$ be the \emph{external field} at vertex $v$.
For each edge $e \in E$, let symmetric matrix $A_e \in \mathbb{R}_{\geq 0}^{q\times q}$ be the \emph{interaction matrix} at edge $e$.
A spin system defines a \emph{Gibbs distribution} $\mu$ over $[q]^V$ such that,
\begin{align*}
\forall \sigma \in [q]^V \quad    \mu(\sigma) \propto w(\sigma) := \prod_{u \in V} b_u(\sigma_u) \prod_{e =\{v,w\} \in E} A_e(\sigma_v,\sigma_w).
\end{align*}
We often use $\Omega(\mu) \subseteq [q]^V$ to denote the support of the Gibbs distribution $\mu$.

Let $\Lambda \subseteq V$ be a subset of vertices. 
Given any \emph{pinning} $\tau \in [q]^{V \setminus \Lambda}$, we define a conditional distribution $\mu^\tau$ by for any configuration $\sigma \in [q]^V$,
\begin{align}\label{eq:def-conditional}
    \mu^\tau(\sigma) \propto w^\tau (\sigma) :=\*1[\sigma_\Lambda =\tau] \cdot \prod_{u \in \Lambda} b_u(\sigma_u) \prod_{\substack{e =\{v,w\} \in E\\ v,w \in \Lambda}} A_e(\sigma_v,\sigma_w) \prod_{\substack{e =\{v,w\} \in E\\ v \in \Lambda\land w \notin \Lambda}} A_e(\sigma_v,\tau_w).
\end{align}
In words, $\mu^\tau$ is a Gibbs distribution obtained for $\mu$  by removing all edges $e \subseteq V \setminus \Lambda$ and putting a constraint that every vertex in $v \in V \setminus \Lambda$ must take the value $\tau_v$. 
In particular, if $\tau$ is feasible (e.g. $\tau$ belongs the support of the marginal distribution $\mu_{V \setminus \Lambda}$), then $\mu^\tau$ is exactly the conditional distribution induced by $\mu$ given the condition $\tau$. 
For all spin systems considered in this paper, it holds that $\sum_{\sigma}w^\tau(\sigma) > 0$ for all $\tau$. The distribution in~\eqref{eq:def-conditional} is well-defined.
Furthermore, for any subset $S$, we use $\mu^\tau_S$ to denote the marginal distribution on $S$ projected from $\mu^\tau$.

The following conditions plays a key role in the proof of our main results. 
Let $\rho: V \to \mathbb{N}_{> 0}$ be a function that maps every vertex $v \in V$ to a positive integer. We call the function $\rho$ the Hamming weight function.
For any two (possibly partial) configurations $\sigma,\tau \in [q]^\Lambda$, where $\Lambda \subseteq V$, define their weighted Hamming distance with respect to $\rho$ by
\begin{align}\label{eq:def-hamming}
    H_\rho(\sigma,\tau)  := \sum_{v \in \Lambda: \sigma(v) \neq \tau(v)} \rho(v).
\end{align}

\begin{definition}[Coupling Independence]\label{def:CI}
Let $C\geq 1$ be a constant.
A distribution $\mu$ over $[q]^V$ is said to be $C$-coupling independent ($C$-CI) 
if there exists Hamming  weight function $\rho:V \to \mathbb{N}_{>0}$ such that the following holds.
For any pinning $\sigma_1,\sigma_2 \in [q]^{S}$, where $S \subseteq V$ and $\sigma_1,\sigma_2$ disagree only at one vertex $v_0 \in S$, there exists a coupling $(X,Y)$, where $X \sim \mu^{\sigma_1}$ and $Y \sim \mu^{\sigma_2}$, such that
\begin{align*}
    \frac{\E[]{H_\rho(X,Y)}}{\rho(v_0)} \leq C.
\end{align*}
\end{definition}

The expectation $\E[]{H_\rho(X,Y)}$ is a kind of Wasserstein distance between $\mu^{\sigma_1}$ and $\mu^{\sigma_2}$.
The notion of coupling independence was introduced explicitly in~\cite{ChenZ23} to study the spectral independence property. 
For example, for Boolean distributions\footnote{The influence matrix and spectral independence are also defined for general distributions with $q \geq 2$. See~\cite{chen2021rapidcolor} for the detailed definition.} ($q = 2$), given any pinning $\tau \in \{-,+\}^{V \setminus \Lambda}$, the  $|\Lambda| \times |\Lambda|$ influence matrix~\cite{anari2020spectral} is defined by
\begin{align}\label{eq:inf-mat}
 \Psi_{\mu}^\tau(v,u) := \mu^{\tau \land v^+}_{u}(+) - \mu^{\tau \land v^-}_u(+),  
\end{align}
where $u,v \in \Lambda$ and $\tau \land v^{\pm}$ denotes the pinning $\tau$ together with $v$ taking the value $\pm$.  
A distribution $\mu$ is $C$-spectrally independent if the maximum eigenvalue of $\Psi_{\mu}^\tau$ is at most $C$ for any pinning $\tau$.
It is not hard to show that  $C$-coupling independence implies $C$-spectral independence.
Hence, recent works~\cite{ChenZ23, chen2023strong,chen2024fast} utilized coupling independence to establish the spectral independence for various spin systems. 



%
\subsection{Compare Markov Chains via Coupling Independence}\label{sec:maintech}
In this work, we find more applications for coupling independence beyond establishing spectral independence. We build some comparison results of Markov chains via coupling independence. As a by-product result, we also show that the coupling independence gives fast sampling algorithms.

Let $\mu$ be a Gibbs distribution over $[q]^V$ on graph $G=(V,E)$. 
For any $\Lambda \subseteq V$, we use $G[\Lambda]$ to denote the induced subgraph of $G$ on vertex set $\Lambda$.

\begin{definition}[Relaxation Time with Pinning]\label{cond:pin-time}
Let $\mu$ be a Gibbs distribution on graph $G=(V,E)$ with maximum degree $\Delta$.
Let $\eta \in [0,1]$.
Let $D(\eta)$ denote all subsets $\Lambda \subseteq V$ such that the maximum degree of $G[\Lambda]$ is at most $\eta \Delta$. Define
\begin{align*}
T_{\textnormal{rel}}^{\textnormal{$(\eta)$}}(\mu) := \max\set{ T_{\textnormal{rel}}^{\textnormal{GD}}(\mu^\tau) \mid \Lambda \in D(\eta) \land \tau \in [q]^{V \setminus \Lambda }  }.    
\end{align*}

In the above definition, $\mu^\tau$ is a distribution on $[q]^V$. In every step, the Glauber dynamics picks $v \in V$ uniformly at random then resamples the value on $v$. If $v \notin \Lambda$, the value of $v$ after resampling is always $\tau_v$. Indeed, $\mu^\tau$ is essentially the same as $\mu^\tau_\Lambda$. But, considering $\mu^\tau$ would help us simplify some results and proofs. The following is our main comparison result.

\end{definition}

\begin{theorem}[Relaxation Time Comparison]\label{thm:main} 
Let $M\geq 1$ and $0 < \eta \leq \frac{1}{2\lceil M \rceil}$ be two constants. There exists $\Delta_0=\Omega(\frac{M^2}{\eta^2}\log \frac{M}{\eta})$ such that for any Gibbs distribution $\mu$ on graph $G$ with the maximum degree $\Delta \geq \Delta_0$, if $\mu$ satisfies $M$-coupling independence, then the relaxation time of Glauber dynamics on $\mu$ satisfies
\begin{align*} T_{\textnormal{rel}}^{\textnormal{GD}}(\mu) \leq 2^{O(M/\eta)}\cdot T_{\textnormal{rel}}^{(\eta)}(\mu).
\end{align*}
\end{theorem}
The theorem is proved in \Cref{sec:relcom}. See \Cref{sec:over} for a proof overview.

The above theorem is a comparison result between two kinds of relaxation times. 
Consider the case when parameters of $\mu$ are close to the critical threshold so that the relaxation time $T_{\textnormal{rel}}^{\textnormal{GD}}(\mu)$ is hard to analyze. By choosing a sufficiently small $\eta$, suppose
for any $\Lambda \in D(\eta)$ and any $\tau \in [q]^{V \setminus \Lambda}$, the conditional distribution $\mu^\tau$ falls into an easy regime.
The relaxation time $T_{\textnormal{rel}}^{(\eta)}(\mu)$ is easy to analyze.
\Cref{thm:main} boosts the relaxation time from an easy regime to the hard regime if $\mu$ satisfies the coupling independence and the maximum degree $\Delta$ is greater than a constant $\Delta_0$.
%
%

When applying \Cref{thm:main} to a specific spin system with Gibbs distribution $\mu$,
we first need to show that the $\mu$ satisfies the coupling independence property. 
Next, we choose a small constant $\eta$ to guarantee that $T_{\textnormal{rel}}^{(\eta)}(\mu)$ is easy to analysis. 
Now, the constant parameter $\Delta_0$ in \Cref{thm:main} is fixed.
If the maximum degree $\Delta \leq \Delta_0 = O(1)$ is bounded, then since the coupling independence implies the spectral independence, the previous work \cite{chen2020optimal} already established the optimal relaxation time for $\mu$. If the maximum degree $\Delta \geq \Delta_0$, we can apply our boosting result to bound the relaxation time.
We show how to prove \Cref{main:coloring} and \Cref{thm:Trel-from-conj} via \Cref{thm:main}.
\paragraph{Proof Sketch of \Cref{main:coloring}}
Given a triangle-free graph $G=(V,E)$ and color lists $L_v \subseteq [q]$ with $|L_v| \geq (\alpha^\star+\delta)\Delta$ for all $v \in V$, let $\mu$ denote the uniform distribution over all proper list-colorings. By going through the analysis in~\cite{feng2021rapid}, we can prove that $\mu$ satisfies $O(1/\delta)$-coupling independence. Let $\eta$ be a parameter to be fixed later. For any $\Lambda \subseteq D(\eta)$, any pinning $\tau \in [q]^{V \setminus \Lambda}$, the distribution $\mu^\tau$ is essentially the same as the distribution $\mu^\tau_\Lambda$ because the coloring outside $\Lambda$ is fixed by $\tau$. By self-reducibility,  $\mu^\tau_\Lambda$ is a list coloring on $G'=G[\Lambda]$ with color list $L'_v = L_v \setminus \{\tau_u \mid u \notin \Lambda \land \{u,v\} \in E\}$. 
Let $\deg'(v)$ and $\deg(v)$ denote the degree of $v$ in $G'$ and $G$ respectively. The new instance satisfies 
\begin{align*}
    \forall v \in \Lambda, \quad \abs{L'_v} \geq \abs{L_v} - (\deg(v) - \deg'(v)) \implies \frac{|L'_v|}{\Delta'} \geq \frac{|L_v|-\deg(v)}{\Delta'},
\end{align*}
where $\Delta'$ denote the maximum degree of $G'$.
By the definition of $D(\eta)$, $\deg'(v) \leq \Delta' \leq \eta \Delta$. We have $|L_v| - \deg(v) > (\alpha^\star-1) \Delta \geq \frac{\alpha^\star - 1}{\eta}\Delta'$. if we set the parameter $\eta \leq \frac{1}{10}$, then
\begin{align}\label{eq:color-good}
    \forall v \in \Lambda, \quad |L'_v| \geq 5 \Delta'.
\end{align}
In this easy regime, one can use path coupling \cite{bubley1997path} to show $T_{\textnormal{rel}}^{(\eta)}(\mu) = O(n)$.
To apply \Cref{thm:main}, we pick a small $\eta = O(\delta)$ and $\eta < \frac{1}{10}$. If $\Delta \geq \Delta_0 = \Theta(\frac{1}{\delta^4} \log \frac{1}{\delta})$, then
\begin{align*}
 T_{\textnormal{rel}}^{(\eta)}(\mu) = 2^{O(1/\delta^2)} n = O_\delta(n).
\end{align*}
On the other hand, if $\Delta \leq \Delta_0 =  \Theta(\frac{1}{\delta^4} \log \frac{1}{\delta})$, then the maximum degree is bounded, we can use the result in~\cite{chen2020optimal} to obtain the relaxation time $T_{\textnormal{rel}}^{(\eta)}(\mu) = O_\delta(n)$ in the same order. This gives the proof sketch of \Cref{main:coloring}.
The only missing component is how to establish the coupling independence, which can be found in \Cref{sec:ListColoring}.

{\paragraph{Proof of \Cref{thm:Trel-from-conj}}
To obtain~\eqref{eq:color-good}, we only need to use the fact that $\alpha^\star > 1$. If we replace $\alpha^\star$ with $(1+\delta)$, then we can set  $\eta \leq \frac{\delta}{5}$ and~\eqref{eq:color-good} still holds.
The same analysis proves \Cref{thm:Trel-from-conj}.
}

\begin{remark}[Hardcore Model in Uniqueness Regime]\label{remark-hard}
\Cref{thm:main} could also rediscover some previous results.
For example, for the hardcore model on a graph $G=(V,E)$ with fugacity $\lambda \leq (1-\delta)\Delta$, \cite{chen2021rapid} proved the optimal $O_\delta(n)$ relaxation time.
For a fixed $\lambda$, the hardcore model falls into an easy regime if we can reduce the maximum degree of the graph by a constant factor.
The hardcore model in the uniqueness regime satisfies $O(1/\delta)$-coupling independence (which can be proved by \Cref{thm:informal} in this paper).
Using a similar argument as that for list coloring, one can rediscover the optimal $O_\delta(n)$ relaxation time using \Cref{thm:main}.     
\end{remark}

%
%
%
%
%

We remark that the relaxation time result for the bipartite graph hardcore model (\Cref{thm:bipart}) is not a direct consequence from \Cref{thm:main}.  
We need to tweak the proof of \Cref{thm:main} to prove \Cref{thm:bipart}. 
The proof of \Cref{thm:bipart} is in \Cref{sec:bihardcore}.
See \Cref{sec:over} for a proof overview.


%
%

\begin{remark}[Compare \Cref{thm:main} to the Technique in~\cite{chen2021rapid}]
Another comparison result about relaxation time was given in~\cite{chen2021rapid}.
The previous result considers general Boolean distribution (not necessarily Gibbs distribution) $\mu$ over $\{-,+\}^V$.
Given a vector $\lambda = (\lambda_v)_{v \in V}$, $(\lambda * \mu)$ denotes the distribution such that for any $\sigma \in \{-,+\}^V$, $(\lambda * \mu) \propto \mu(\sigma)\prod_{v \in V: \sigma(v) = +}\lambda_v$. The result  says if $(\lambda * \mu) $ is spectrally independent for all $\lambda \in (0,1]^V$, then one can compare $T_{\textnormal{rel}}^{\textnormal{GD}}(\mu)$ to $T_{\textnormal{rel}}^{\textnormal{GD}}(\lambda_\theta*\mu)$, where $\lambda_\theta$ is the vector with constant value $0< \theta < 1$.
When applying results to Gibbs distributions, here are some differences between \Cref{thm:main} and the previous result. 
\begin{itemize}
    \item \Cref{thm:main} works for general domain $[q]$ but previous result works only for Boolean domain;
    \item The condition is incomparable. \Cref{thm:main} requires coupling independence for $\mu$ and a degree lower bound for the underlying graph but the previous result requires spectral independence for a family of distributions;
    \item The easy regime is incomparable. The easy regime in \Cref{thm:main} is the conditional distributions on a small degree subgraph but the easy regime in the previous result is $\lambda_\theta*\mu$;
\end{itemize}
\end{remark}

For many spin systems, one can use \Cref{thm:main} to establish the optimal $O(n)$ relaxation for Glauber dynamics, where $n$ is the number of variables in the spin system. 
By the standard relation between mixing and relaxation time, the mixing time of Glauber dynamics can usually be bound by $O(n^2)$. 
Each transition of Glauber dynamics can be simulated in time $O(\Delta)$. Hence, one can obtain a sampling algorithm in time $O(\Delta n^2)$. 
Alternatively, we can give a faster sampling algorithm in time $\widetilde{O}(\Delta n)$ if the easy regime has linear-near mixing time.

\begin{definition}[Mixing Time with Pinning]\label{cond:pin-time-mix}
Let $\mu$ be a Gibbs distribution on graph $G=(V,E)$ with maximum degree $\Delta$.
Let $\eta \in [0,1]$.
Let $D(\eta)$ denote all subsets $\Lambda \subseteq V$ such that the maximum degree of $G[\Lambda]$ is at most $\eta \Delta$. Define
\begin{align*}
T_{\textnormal{mix}}^{(\eta)}(\mu) := \max\set{ T_{\textnormal{mix}}^{\textnormal{GD}}\left(\mu^\tau,\frac{1}{4e}\right) \mid \Lambda \in D(\eta) \land \tau \in [q]^{V \setminus \Lambda }  }.   
\end{align*}
\end{definition}
In words, for any pinning $\tau$ on $V \setminus \Lambda$ with $\Lambda \in D(\eta)$,  $T_{\textnormal{mix}}^{(\eta)}(\mu)$ is an upper bound for the mixing time $T$ of Glauber dynamics for $\mu^\tau$ such that starting from the worst initial $X_0$, the total variation distance between $X_T$ and $\mu^\tau$ is at most $\frac{1}{4e}$.


\begin{theorem}[Fast Sampling Algorithm]\label{thm:algo}
Let $M\geq 1$ and  $0 < \eta \leq \frac{1}{2\lceil M \rceil}$ be two constants. 
There exists an algorithm such that given any $\epsilon \in (0, 1)$ and  any Gibbs distribution $\mu$ on graph $G$ with the maximum degree $\Delta \geq \Delta_0 = (\frac{10M}{\eta})^2\log \frac{10M}{\eta}$, if $\mu$ satisfies $M$-coupling independence such that the weighted hamming distance $\rho$ satisfies $\frac{\rho_{\max}}{\rho_{\min}} = \-{poly}(n)$,
then it returns a random sample $X$ satisfying $\DTV{X}{\mu} \leq \epsilon$ in time \[  \Delta T_{\textnormal{mix}}^{{(\eta)}}(\mu) 
 \tp{\log \frac{n}{\epsilon}}^{O(M/\eta)},\]
 where  $n$ is the number of vertices in $G$ and we use $\rho_{\max} = \max_{v\in V}\rho(v)$ and $\rho_{\min} = \min_{v \in V} \rho(v)$.
\end{theorem}
The theorem is proved in \Cref{sec:algo}. See \Cref{sec:over} for a proof overview.

In the above theorem, suppose $\frac{\rho_{\max}}{\rho_{\min}} = O(n^d)$ for some universal constant $d$, then the running time in above theorem should be $T_{\textnormal{mix}}^{\textnormal{GD}}(\mu,\eta) \Delta (\log \frac{n}{\epsilon})^{C(M/\eta + d)}$ for some universal constant $C$.
We then hide the constants $C$ and $d$ by $O(\cdot)$ in \Cref{thm:algo}.

\Cref{main:coloring-algo} can be obtained from \Cref{thm:algo}.
Consider the list coloring on a triangle free graphs $G=(V,E)$ with $|L_v| \geq (\alpha^\star+1)\Delta$. 
The uniform distribution $\mu$ satisfies $O(1/\delta)$-coupling-independence with standard Hamming weight $\rho(v) = 1$ for all $v \in V$.
Take $\eta = O(1/\delta)$ be a small constant with $\eta < \frac{1}{10}$.
By~\eqref{eq:color-good}, a simple path coupling~\cite{bubley1997path} shows that $T_{\textnormal{mix}}^{\textnormal{GD}}(\mu,\eta) = O(n \log n)$. 
Hence, if the maximum degree of $\Delta$ is greater than $\Delta_0$, we run the algorithm in \Cref{thm:algo} and the running time is $\Delta n \cdot \mathrm{polylog}(\frac{n}{\epsilon})$. Otherwise, the maximum degree is bounded and the result in~\cite{chen2020optimal} gives the $O_\delta(n \log n)$ mixing time of the Glauber dynamics, then we can simulate Glauber dynamics to obtain a sampling algorithm.
The proof of \Cref{main:coloring-algo} is in \Cref{sec:ListColoring}.

The algorithm in \Cref{thm:algo} is \emph{not} the Glauber dynamics. Roughly speaking, the algorithm uses some strategy to pick vertex and uses the  Glauber update to resample the value of the picked vertex. 
However, for \emph{monotone spin systems}, we can compare this algorithm to Glauber dynamics via censoring inequality~\cite{PW13} and then we can bound the mixing time of Glauber dynamics.


Let $\mu$ over $[q]^V$ be the Gibbs distribution.
Define a partial order $\leq$ for $[q]^V$ as follows. For each $v \in V$, pick a \emph{total order} $\leq_v$ on $[q]$.
For any two $X, Y \in [q]^V$,
\begin{align} \label{eq:def-leq}
  X \leq Y \quad \Longleftrightarrow \quad \forall v\in V, \quad X_v \leq_v Y_v.
\end{align}
For two distributions $\mu$ and $\nu$ over $[q]^V$, we say \emph{$\mu$ is stochastic dominated by $\nu$} (i.e., $\mu \preceq \nu$) if there is a coupling $\+C$ between $\mu, \nu$ such that $\Pr[(X,Y) \sim \+C]{X \leq Y} = 1$.
Let $P$ be the transition matrix of the Glauber dynamics on $\mu$, which can be written as
\begin{align} \label{eq:GD-decomp}
  P = \frac{1}{n} \sum_{v \in V} P_v,
\end{align}
where $P_v$ performs \emph{updates} at the $v \in V$ such that $P_v(X, Y) = \mu^{X_{V \setminus v }}(Y)$, for all $X, Y \in [q]^V$.
\begin{definition}\cite[Chapter 22]{levin2017markov} \label{def:monotone-spin-system}
We say $\mu$ is a \emph{monotone spin system} if 
for every $v \in V$, $P_v$ is ordering persists, which means for any $X, Y \in [q]^V$ with $X \leq Y$, it holds that $P_v(X, \cdot) \preceq P_v(Y, \cdot)$.
\end{definition}

By the definition of the partial ordering $\leq$ in $[q]^V$, there is a unique maximum configuration for the ordering. Denote this state as $X^+$. Recall $T_{\textnormal{mix}}^{\textnormal{GD}}(\mu, \epsilon \mid X^+)$ denotes the mixing time of Glauber dynamics starting from $X^+$.

\begin{theorem}[Mixing Time Comparison]\label{thm:mixcom}
Let $M\geq 1$ and  $0 < \eta \leq \frac{1}{2\lceil M \rceil}$ be two constants. 
For any monotone spin system $\mu$ on graph $G$ with the maximum degree $\Delta = \Omega(\frac{M^2}{\eta^2}\log \frac{M}{\eta})$, if $\mu$ satisfies $M$-coupling-independence such that the Hamming weight $\rho$ satisfies $\frac{\rho_{\max}}{\rho_{\min}} = \mathrm{poly}(n)$,
then the mixing time of Glauber dynamics starting from the maximum configuration satisfies
\begin{align*} 
T_{\textnormal{mix}}^{\textnormal{GD}}(\mu, \epsilon \mid X^+) \leq \tp{\log \frac{n}{\epsilon}}^{O(M/\eta)}\cdot T_{\textnormal{mix}}^{{(\eta)}}(\mu),
\end{align*}
where $n$ is the number of vertices in graph $G$.
\end{theorem}
The theorem is proved in \Cref{sec:mixcom}. See \Cref{sec:over} for a proof overview.

The above theorem is of independent interest. 
Suppose $\mu$ is a monotone system with coupling independence property.
The parameters of $\mu$ are in the critical regime and the underlying graph has an unbounded maximum degree.
If we can choose a proper constant $\eta$ such that $T_{\textnormal{mix}}^{{(\eta)}}(\mu) = O(n \log n)$, then the theorem gives a linear-optimal $\widetilde{O}(n)$ mixing time of Glauber dynamics for $\mu$ starting from the maximum configuration. 
%

To obtain the linear-near mixing time for $\mu$, some previous works~\cite{Chen0YZ22,ChenE22,AnariJKPV22}  developed comparison techniques for the modified log-Sobolev (MLS) constants.
Roughly speaking, if one can lower bound the MLS constant $\text{mls}(\mu)$ of the Glauber dynamics for $\mu$, then one can obtain the optimal $O(n \log n)$ mixing time.
Previous works compared $\text{mls}(\mu)$ to $\text{mls}(\mu')$, where $\mu'$ is a distribution in the easy regime, and such comparison requires $\mu$ to satisfy certain entropic independence~\cite{AnariJKPV22} condition. 
In general, it is not easy to verify the entropic independence condition and analyze $\text{mls}(\mu')$ even if $\mu'$ is in an easy regime. 
\Cref{thm:mixcom} only requires the coupling independence condition and directly compares the mixing time.
However, \Cref{thm:mixcom} requires monotone systems, and the final mixing result is restricted.

We remark that although the hardcore model in bipartite graphs is a monotone system, \Cref{thm:bipartmix} is not a direct consequence from \Cref{thm:mixcom}.  
We need to tweak the proof of \Cref{thm:mixcom} to prove \Cref{thm:bipartmix}. 
The proof of \Cref{thm:bipartmix} is in \Cref{sec:bihardcore}.
See \Cref{sec:over} for a proof overview.
%


\subsection{Establish Coupling Independence}

The next question is how to establish the coupling independence condition for spin systems. 
Previously, spectral independence was known for many spin systems.
The coupling independence was often a by-product result when proving spectral independence.
Hence, it is known for some specific spin systems such as subgraph world~\cite{ChenZ23}, $b$-matching~\cite{chen2024fast} and coloring in high girth graphs~\cite{chen2023strong}.

In this paper, we give a tool to turn many existing spectral independence results into coupling independence results. 
A large family of spin systems is $2$-spin systems.
Let $G = (V, E)$ be a graph with maximum degree $\Delta \geq 3$.
Let $0 \leq \beta \leq \gamma$ be the edge interactions such that $\gamma > 0$.
Let $\lambda > 0$ be the external field.
Let $\mu$ be the Gibbs distribution on $G$ with parameters $\beta, \gamma, \lambda$ such that for any $\sigma \in \{-,+\}$, $\mu(\sigma) \propto \lambda^{n_+(\sigma)} \beta^{m_+(\sigma)}\gamma^{m_-(\sigma)}$, where $n_+(\sigma)$ is the number of vertices $v$ with $\sigma_v = +$ and $m_{\pm}(\sigma)$ is the number of edges $\{u,v\}$ with $\sigma_u = \sigma_v = \pm$. 
The 2-spin system is said to be ferromagnetic if $\beta \gamma > 1$ and anti-ferromagnetic if $\beta \gamma < 1$. 
%

Anari, Liu, and Oveis Gharan~\cite{anari2020spectral} analyzed the spectral independence for the hardcore model.
Chen, Liu, and Vigoda~\cite{chen2020rapid} extended the analysis to general 2-spin systems.
Recall the influence matrix $\Psi_{\mu}^\tau$ is defined in~\eqref{eq:inf-mat}. 
The maximum eigenvalue $\mathrm{Eig}_{\max}(\Psi_{\mu}^\tau)$ can be upper bound by the total influence from one vertex
\begin{align}\label{eq:infbound}
    \mathrm{Eig}_{\max}(\Psi_{\mu}^\tau) \leq \max_{v} \sum_{u \in V} |\Psi_{\mu}(v,u)|. 
\end{align}
The RHS is called the \emph{total influence bound}.
For 2-spin systems,
the analysis is performed on the Self-Avoiding-Walk (SAW) tree~\cite{weitz2006counting}.
Roughly speaking, fix a vertex $v$, the SAW tree $T_v$ enumerates all the SAWs in graph $G$ starting from $v$.
By defining a proper 2-spin system on $T_v$, one can use the total influence from the root in $T_v$ to upper bound the total influence from $v$ in $G$, and thus establish the spectral independence for Gibbs distribution $\mu$.
In~\cite{chen2020rapid}, a weighted version of~\eqref{eq:infbound} is studied to deal with general 2-spin systems.
We give the following result for coupling independence.

\begin{theorem}[Informal version of \Cref{lem:CI-tool}]\label{thm:informal}
For 2-spin systems, the (weighted) total influence bound in the Self-Avoid-Walk tree implies coupling independence.    
\end{theorem}

As a consequence, all the spectral independence results for 2-spin systems in~\cite{chen2020rapid} can be turned into coupling independence results in black-box.
For the hardcore model in bipartite graphs( \Cref{thm:bipart} and \Cref{thm:bipartmix}), we can also use the above result to transform the total influence bound in~\cite{chen2023uniqueness} into coupling independence result. 

\Cref{thm:informal} is proved by constructing a recursive coupling in \Cref{sec:CI}. Fix a vertex $v$ in $G$. We build a coupling $(X,Y)$ between $\mu^{v^+}$ and $\mu^{v^-}$ and show the discrepancy between $X$ and $Y$ are bounded by the total influence in the SAW tree $T_v$. 
Suppose $v$ has $d$ neighbors $u_1,u_2,\ldots,u_d$.
We split $v$ into $d$ copies $v_1,v_2,\ldots,v_d$ such that $v_i$ only has one neighbor $u_i$.
Define the pinning $\sigma_i$ such that $v_j$ for $j \geq i$ takes the value $+$ and $v_j$ for $j <i$ takes the value $-$.
Then $\mu^{v^+} = \mu^{\sigma_0}$ and $\mu^{v^-} = \mu^{\sigma_d}$. 
We couple each adjacent $\mu^{\sigma_{i-1}}$ and $\mu^{\sigma_i}$, then merge them into a coupling between two endpoints $ \mu^{\sigma_0}$ and  $\mu^{\sigma_d}$.
For each adjacent pair, the only difference between $\sigma_i$ and $\sigma_{i-1}$ is the pinning at $v_i$.
Hence, we first couple the only neighbor $u_i$ of $v_i$ then construct the coupling recursively if the coupling at $u_i$ fails.
This recursive processing essentially enumerates all SAWs from $v$. 
We can relate the coupling with the SAW tree to prove the theorem.

For multi-spin systems such as list-coloring, we can mimic the recursive coupling for 2-spin systems.
Since the previous spectral independence results for list-coloring were also obtained via the SAW tree~\cite{feng2021rapid,chen2021rapidcolor},
a similar proof gives the coupling independence.


\subsection{Proof Overview}\label{sec:over}

We give a proof overview for the relaxation time comparison result in \Cref{thm:main}. 
Let $G=(V,E)$ be a graph with maximum degree $\Delta$.
Let $\ell$ and $k$ be two constant integers such that $\ell < k$. Their specific values will be fixed later.
We first partition all the vertices in $V$ into $k$ parts $U_1,U_2,\ldots,U_k$ such that for any vertex $v \in V$, each $U_i$ does not have more than $\frac{\eta}{\ell} \Delta$ neighbors of $v$. In other words, let $\Gamma_v = \{ u \mid (u,v) \in E\}$ denote the set of neighbors of $v$ in graph $G$. For any $i \in [k]$, $|\Gamma_v \cap U_i| \leq  \frac{\eta}{\ell} \Delta$.
The existence of the partition is guaranteed by the Lov{\'a}sz local lemma. However, the local lemma requires the maximum degree $\Delta$ to be sufficiently large. That is why we require a lower bound for $\Delta$ in our technical results.
We also remark that in our proof, the degree lower bound is used solely to ensure the existence of the partition.
A similar partition appeared in the previous work \cite{JainSS21}.

The input Gibbs distribution $\mu$ over $[q]^V$ is a joint distribution of $n$ variables $(X_v)_{v \in V}$, where each variable takes its value from $[q]$. 
Now, we can view $\mu$ as a joint distribution of $k$ variables $Y = (Y_i)_{i \in [k]}$ such that each variable $Y_i = X_{U_i}$ takes its value from a huge domain $[q]^{U_i}$. 
We define the $k \leftrightarrow (k - \ell)$ down-up walk on $Y$. Given $Y = (Y_1,Y_2,\ldots,Y_k)$, the Markov chain does as follows 
\begin{itemize}
    \item Down-walk: Sample a set $S \in \binom{[k]}{\ell}$ of $\ell$ indices uniformly at random and then remove the configuration on the set $S$: $Y \to Y_{[k] \setminus S}$;
    \item Up-walk: Resample $Y_S$ from $\mu$ conditional on $Y_{[k] \setminus S}$ and then go back to a full configuration $Y_{[k] \setminus S} \to Y_{[k] \setminus S} \cup Y_S$.
\end{itemize}
A full configuration $Y=(Y_1,Y_2,\ldots,Y_k)$  is on the level $k$. In the down-walk, we sample a random subset of indices $S \subseteq [k]$ with size $\ell$. By dropping the configuration $Y_S$, we move from a full configuration at level $k$ to a partial configuration at level $k - \ell$. In the up-walk, we resample $Y_S$ and go back to the level $k$.
The process can also be viewed as a kind of block dynamics for configuration $X \in [q]^V$.
In every step, we pick a random subset $U_S = \cup_{i \in S}U_i \subseteq V$ of variables and resample $X(U_S)$ conditional on $X(V \setminus U_S)$. 

We use local-to-global technique~\cite{alev2020improved,anari2020spectral} to analyze the spectral gap of the $k \leftrightarrow (k-\ell)$ down-up walk for $Y$.
The local-to-global technique suggests to analyze the relaxation time of $k \leftrightarrow 1$ down-up walk\footnote{In \cite{alev2020improved,anari2020spectral}, the local walk is essentially defined as the $1 \leftrightarrow k$ up-down walk. Every state is $Y_i$ for $i \in [k]$. In the up-walk, it extends $Y_i$ to a full configuration $Y$. In the down-walk step, it samples a random index $j \in [k]$ and updates $Y$ to $Y_j$. It is well-known that  $1 \leftrightarrow k$ up-down walk and $k \leftrightarrow 1$ up-down walk has the same relaxation time.}. In the down walk, we pick a random $S$ of size $|S| = k - 1$ and drop $Y_S$. In the up-walk, we resample $Y_S$ and go back to level $k$. 
We use coupling independence to analyze this  $k \leftrightarrow 1 $ down-up walk via path coupling. 
For simplicity, suppose $\mu$ satisfies $C$-coupling independence with standard Hamming distance ($\rho(v) = 1$ for all $v \in V$).
We can view this $k \leftrightarrow 1 $ down-up walk on $Y$ as a block dynamics on $X \in [q]^V$, where it updates a block $U_S$ in every step. 
Given two $X \in [q]^V$ and $X' \in [q]^V$ that disagree only at one vertex $v \in V$, say $v \in U_1$,  we couple the transition of $k \leftrightarrow 1$ down-up walk. 
Let two $k \leftrightarrow 1$ down-up walks (starting from $X$ and $X'$, respectively) select the same random subset $S \subseteq [k]$ such that $\abs{S} = k-1$.
\begin{itemize}
    \item If $1 \in S$, which happens with probability $\frac{k-1}{k}$, then since $v \in U_1$ the value of $v$ is removed in the down-walk, and thus $X$ and $X'$ can be coupled perfectly after the transition.
    \item If $1 \notin S$, which happens with probability $\frac{1}{k}$, then since $v \in U_1$, the disagreement at $v$ may percolate to other blocks in the up-walk step. We use the coupling in the $M$-coupling independence to couple the up-walk so that the expected Hamming distance between $X$ and $X'$ after the transition is at most $M$. 
\end{itemize}
Hence, the expected Hamming distance between $X$ and $X'$ after transition is at most $\frac{M}{k}$. 
If $k> M$, the path coupling gives the $O(\log n)$ mixing time and $O(1)$ relaxation time for this down-up walk. 
To apply the local-to-global technique, we also need to fix a configuration $Y_\Lambda$, where $\Lambda \subseteq [k]$ and $ |\Lambda| = t \leq \ell$, and consider the $(k-t) \leftrightarrow 1$ down-up walk for $Y_{[k] \setminus \Lambda}$. The same path coupling works if $k - t > M$. 
By choosing $k$ and $\ell$ such that $k - \ell > M$ and using the local-to-global technique,
we can show that the $k \leftrightarrow (k-\ell)$ down-up walk for $Y$ has $O(1)$ relaxation time. 

We then compare the $k \leftrightarrow (k-\ell)$ down-up for $Y=(Y_1,Y_2,\ldots,Y_k)$ to the Glauber dynamics for $X \in [q]^V$.
Recall that $k \leftrightarrow (k-\ell)$ down-up walk is a kind of block dynamics for $X$.
In every step, the block dynamics updates a subset $U_S = \cup_{i \in S}U_i$ with $|S| = \ell$. The update step is to resample $X(U_S)$ conditional on $X(V \setminus U_S)$. This step samples from the conditional Gibbs distribution $\mu^{X({ V \setminus U_S }) }_{U_S}$ on subgraph $G[U_S]$.
By the construction of the partition, the maximum degree of $G[U_S]$ is at most $\eta \Delta$ so that we have the relaxation time bound $T^{(\eta)}_{\text{rel}}(\mu)$ for Glauber dynamics on $\mu^{X_{ V \setminus U_S } }_{U_S}$. 
Let $T^{\text{down-up}}_{\text{rel}}$ denote the relaxation time of $k \leftrightarrow (k-\ell)$ down-up walk.
By some standard comparison argument between block dynamics and the Glauber dynamics, we can prove \Cref{thm:main} by showing that 
\begin{align*}
    T^{\text{GD}}_{\text{rel}}(\mu) \leq T^{\text{down-up}}_{\text{rel}} \times T^{(\eta)}_{\text{rel}}(\mu) = O(1) \times T^{(\eta)}_{\text{rel}}(\mu).
\end{align*}

Next, we briefly explain how to get the near-linear time sampling algorithm in \Cref{thm:algo} and the mixing time in \Cref{thm:mixcom}.
Note that for $k \leftrightarrow 1$ down-up walk, the path coupling actually gives the $O(\log n)$ mixing time.
For one update step, it selects a subset $S \subseteq [k]$ with $|S| = k - 1$. Let $i$ denote the missing index, i.e. $S \cup \{i\} = [k]$. The update step resamples $Y_S$ conditional on $Y_i$. 
We can simulate this transition step using $(k-1) \leftrightarrow 1$ down-up walk for the conditional distribution on $Y_S$.
This down-up walk also has the $O(\log n)$ mixing time.
We do this recursively until we need to sample from a conditional distribution on $Y_{S'}$ with $|S'| = \ell$. 
Note that the maximum degree of the graph $G[U_{S'}]$ is at most $\eta \Delta$.
Now, we simulate the Glauber dynamics for $T^{(\eta)}_{\text{mix}}(\mu)$ steps to sample from the conditional distribution.
Hence, the total number of Glauber steps is $(\log n)^{O(\ell)}T^{(\eta)}_{\text{mix}}(\mu)$.
For monotone systems, we can compare this algorithm to Glauber dynamics via censoring inequality.

The results for list-coloring are consequences of general technical results.
However, we need to tweak the analysis to prove the results for hardcore model in bipartite graphs (\Cref{thm:bipart} and \Cref{thm:bipartmix}). 
The reason is that for hardcore model on $G=(V_L \cup V_R, E)$, we only know $\lambda < \lambda_c(\Delta_L)$ but we cannot control the degree $\Delta_R$ in the right part $V_R$. 
Our technique can only prove the coupling independence for $\mu_L$, which is the marginal distribution on $V_L$ projected from $\mu$.
To prove the relaxation time and mixing time results, we first partition $V_L$ into disjoint part $U_1,U_2,\ldots,U_k$ such that for any vertex $v \in V_R$, $v$ has no more than $\frac{\eta}{\ell}\Delta$ neighbors in each $U_i$.
Again, the existence of the partition is guaranteed by the local lemma.
Let $X \sim \mu_L$ be a partial configuration on $V_L$. 
We can define $Y= (Y_1,Y_2,\ldots,Y_k)$, where $Y_i = X_{U_i}$.
By a similar proof, we show that the $k \leftrightarrow (k-\ell)$ down-up walk for $Y$  is rapid mixing. 
We consider a global Markov chain over $ \{-,+\}^{V_L \cup V_R}$ defined as follows. Let $\overline{X} \in \{-,+\}^{V_L \cup V_R}$ be a full configuration.
\begin{itemize}
    \item Drop the right part to obtain $X \gets \overline{X}({V_L})$;
    \item Update $X$ using the $k \leftrightarrow (k-\ell)$ down-up walk for $Y = (Y_1,\ldots,Y_k)$, where $Y_i = X({U_i})$;
    \item Sample $X({V_R}) \sim \mu_{V_R}^{X}$ and  let $\overline{X} \gets X \cup X({V_R})$.
\end{itemize}
We first compare this chain to the $k \leftrightarrow (k-\ell)$ down-up walk and then compare the Glauber dynamics for $\mu$ to this Markov chain. This gives the relaxation time of Glauber dynamics. For the mixing time, we can first obtain a near-linear time sampling algorithm for $\mu_L$, since hardcore model in bipartite graph is a monotone system, we then compare the algorithm to the Glauber dynamics for $\mu$ via censoring inequality.




\section{Preliminaries}\label{sec:prelim}

\subsection{\texorpdfstring{$\phi$}{}-Divergences and \texorpdfstring{$\phi$}{}-Entropies}
Let $\phi: D \to \^R$ be a convex function with domain $D \subseteq R$.
Let $\mu$ be a distribution over a finite set $\Omega$.
For any random variable $f:\Omega \to D$, the $\phi$-entropy of $f$ with respect to $\mu$ is defined as
\begin{align} \label{eq:def-phi-ent}
    \PhiEnt[\mu]{f} := \E[\mu]{\phi(f)} - \phi(\E[\mu]{f}).
\end{align}
Note that $\PhiEnt[\mu]{f} \geq 0$ follows directly from the Jensen's inequality since $\phi$ is a convex function.
In particular, when $D \supseteq R_{\geq 0}$, the notion of $\phi$-entropy can be used to measure the distance between distributions.
Let $\nu$ and $\mu$ be distributions on $\Omega$ such that $\nu$ is absolutely continuous with respect to $\mu$, then the $\phi$-divergence $\+D_\phi(\nu \parallel \mu)$ between $\nu$ and $\mu$ is defined as
\begin{align} \label{eq:def-phi-divergence}
    \+D_\phi(\nu \parallel \mu) &:= \PhiEnt[\mu]{\frac{\nu}{\mu}}.
\end{align}
In practice, typical choice of the function $\phi$ are 
\begin{itemize}
    \item $\text{TV}(x) = \frac{1}{2}\abs{x - 1}$: this defines the \emph{TV-distance} $\DTV{\nu}{\mu} = \frac{1}{2}\sum_{x\in \Omega}\abs{\nu(x) - \mu(x)}$;
    \item $\chi^2(x) = x^2$: this defines the \emph{$\chi^2$-divergence} $\+D_{\chi^2}(\nu \parallel \mu) = (\sum_{x \in \Omega} \nu(x) \cdot \frac{\nu(x)}{\mu(x)}) - 1$;
    \item $\text{KL}(x) = x \log x$: this defines the \emph{KL-divergence} $\DKL{\nu}{\mu} = \sum_{x \in \Omega} \nu(x) \log \frac{\nu(x)}{\mu(x)}$.
\end{itemize}
By conventions, the $\chi^2$-entropy is usually called variance (with the notation $\Var[\mu]{f}$) and the $\text{KL}$-entropy is usually called entropy (with the notation $\Ent[\mu]{f}$):
\begin{align}
    \label{def:var}
    \Var[\mu]{f} &= \E[\mu]{f^2} - \E[\mu]{f}^2 \\
    \label{def:ent}
    \Ent[\mu]{f} &= \E[\mu]{f \log f} - \E[\mu]{f} \log \E[\mu]{f}.
\end{align}

\subsection{Markov Chain Background}

\begin{definition}
    Let $\Omega$ and $\Omega'$ be two finite sets.
    A \emph{Markov kernel} $P$ from $\Omega$ to $\Omega'$ assigns to every element $x \in \Omega$ a distribution $P(x, \cdot)$ on $\Omega'$.
    In particular, $P$ could be seen as a matrix in $\^R^{\Omega \times \Omega'}_{\geq 0}$.
\end{definition}

When $\Omega = \Omega'$, the Markov kernel $P$ becomes the \emph{transition matrix} of some Markov chain $(X_t)_{t \geq 0}$.
We use $P$ to refer to this Markov chain if it is clear in the context.
The Markov chain $P$ is
\begin{itemize}
    \item \emph{irreducible}, if for any $x, y \in \Omega$, there is a $t > 0$ such that $P^t(x, y) > 0$;
    \item \emph{aperiodic}, if for any $x \in \Omega$, it holds that $\-{gcd}\set{x \mid P^t(x, x) > 0} = 1$.
\end{itemize}
A distribution $\mu$ on $\Omega$ is called the stationary distribution of $P$ if $\mu P = \mu$.
If a Markov chain is both irreducible and aperiodic, then it has a unique stationary distribution.

\begin{definition}[time-reversal]
    Let $\Omega$ and $\Omega'$ be two finite sets.
    The \emph{time-reversal} $P^*$ of a Markov kernel $P$ from $\Omega$ to $\Omega'$ with respect to $\mu$ is defined by the following \emph{detailed balanced equation}:
    \begin{align} \label{eq:detailed-balanced-ineq}
        \forall x \in \Omega, y \in \Omega', \quad \mu (x) P(x, y) = \mu^*(y) P^*(y, x),
    \end{align}
    where $\mu^* = \mu P$ is the corresponding distribution on $\Omega'$.
\end{definition}

In particular, let $P$ be a Markov chain on $\Omega$.
Let $P^*$ be its time-reversal with respect to $\mu$.
Then $P$ is called \emph{reversible} with respect to $\mu$ if $P = P^*$.
This implies $\mu$ is a stationary distribution of $P$.


Let $P$ be a Markov chain on $\Omega$ with the unique stationary distribution $\mu$.
The mixing time of Glauber dynamics is defined in~\eqref{eq:def-mix}. It can be defined similarly for a general Markov chain $P$. 

In this paper, we are particularly interested in the block dynamics defined as follows.

\paragraph{Block dynamics}
Let $\mu$ be a distribution over $[q]^V$, not necessarily Gibbs distribution. Let $\+B=\{B_1,B_2,\ldots,B_\ell\}$ a set of blocks, where $B_i \subseteq V$. Define the following block dynamics. Given any $X \in \Omega(\mu)$, the block dynamics does as follows 
\begin{itemize}
    \item down-walk $D$: sample $i \in [\ell]$ uniformly at random and let $X \to X_{V \setminus B_i}$;
    \item up-walk $U$: sample $X_{B_i} \sim \mu_{B_i}^{X_{V \setminus B_i}}$ and extend $X_{V \setminus B_i} \to X_{B_i} \cup X_{V \setminus B_i}$. 
\end{itemize}
Here, we decompose the block dynamics into two steps, the down-walk $D$ and the up-walk $U$.
Let $\Omega = \Omega(\mu)$ be the set of full configurations.
Let $\Omega^* = \{X_{V \setminus B_i } \mid X \in \Omega \land 1\leq i \leq \ell \}$ be a set of partial configurations.
The down-walk $D: \Omega \times \Omega^* \to \mathbb{R}_{\geq 0}$ goes from a full configuration to a random partial configuration. 
And the up-walk $U: \Omega^* \times \Omega \to \mathbb{R}_{\geq 0}$ goes from a partial configuration to a random full configuration.
We view both $D$ and $U$ as transition matrices.
Moreover, $U = D^*$ is the time-reversal of $D$ with respect to $\mu$.
Let $P = DU$ be the composition of $D$ and $U$.
$P$ is the transition matrix of the block dynamics and $P$ is time-reversible with respect to $\mu$. In particular, if each $B_i$ only contains a single variable, then the block dynamics above is exactly the Glauber dynamics. 

    %

In order to investigate the relaxation time of the block dynamics.
We primarily use the notion of approximate tensorization~\cite{MSW03, CMT15} and block factorization of variance~\cite{CP20}.
Recall $\mu$ is a distribution over $[q]^V$.
We use $\Omega(\mu)$ to denote the support of $\mu$.
Let $f: \Omega(\mu) \to \mathbb{R}$ be a function.
And $X \sim \mu$ be a random sample.
We will use $\E[\mu]{f}$ to denote $\E{f(X)}$.
Recall that the variance of $f$ is the $\chi^2$-entropy of $f$ defined as $\Var[\mu]{f} = \Var{f(X)} = \E{f^2(X)} - \E{f(X)}^2$. 

For any subset $S \subseteq V$, define 
\begin{align*}
\Var[S]{f} &:= \Var{f(X) \mid X_{V\setminus S}} \\
\text{and} \quad \mu[\Var[S]{f}] &:= \E{\Var{f(X) \mid X_{V\setminus S}}}
\end{align*}
as a shorthand. 
We note that $\Var[S]{f}$ is a random variable $\Omega(\mu) \to \^R$.
Given a sample $x \in \Omega(\mu)$, the value of this random variable is given by
\begin{align*}
    \oVar^x_S[f] &:= \Var{f(X) \mid X_{V\setminus S}}(x)
    = \Var{f(X) \mid X_{V\setminus S} = x_{V\setminus S}}.
\end{align*}
\begin{definition} \label{def:BF}
   Let $\+B=\{B_1,B_2,\ldots,B_\ell\}$ a set of blocks, where $B_i \subseteq V$.
   We say $\mu$ satisfies \emph{$\+B$-block factorization of variance} with parameter $C$ if the following inequality holds for every $f:\Omega(\mu) \to \^R$:
\begin{align*}
    \Var[\mu]{f} \leq \frac{C}{\ell}\sum_{B \in \+B}\mu[\Var[B]{f}].
\end{align*}
In particular, if every $B_i$ is a single variable $v \in V$ and $\ell = \abs{V}$, then we say $\mu$ satisfies \emph{$C$-approximate tensorization of variance}.
\end{definition}

The block factorization of variance is closely related to the contraction of the down-walk $D$ on the $\chi^2$-divergence.
We formally define it as follow.

\begin{definition}
We say the down-walk of the block dynamics has \emph{$\delta$-contraction for $\chi^2$-divergence} if for any distribution $\nu$ over $\Omega(\mu)$ such that $\nu$ is absolutely continuous with respect to $\mu$,
\begin{align*}
    \+D_{\chi^2}( \nu D \parallel \mu D) \leq (1 - \delta) D_{\chi^2}( \nu \parallel \mu).
\end{align*}
\end{definition}


It is well-known that the relaxation time, block factorization and the contraction of the down-walk are equivalent.
We summarize this into the following proposition.
\begin{proposition} \label{prop:relax-tensor-contract}
The following three properties are equivalent for block dynamics:
\begin{enumerate}
    \item the relaxation time of block dynamics is $T$;
    \item $\mu$ satisfies $\+B$-block factorization of variance with parameter $T$;
    \item the down-walk has $1/T$-contraction for $\chi^2$-divergence.
\end{enumerate}
\end{proposition}
For the equivalence of Item 1 and Item 2, we refer the readers to~\cite{CMT15,Cap23}.
And for the equivalence of Item 2 and Item 3, we refer the readers to~\cite[Lemma 2.7]{chen2020optimal}.

\subsection{Coupling and Path Coupling}
Let $(X, Y)$ be two jointly distributed random variables such that $X \sim \nu$ and $Y \sim \mu$.
Then, $(X, Y)$ is called a \emph{coupling} of $\nu$ and $\mu$.
The TV-distance between two distribution $\nu$ and $\mu$ over $\Omega$ can be bounded by coupling.
The following result is standard and is called the coupling lemma.
We refer the readers to~\cite[Proposition 4.7]{levin2017markov} for a proof.
\begin{lemma}[coupling lemma] \label{lem:coupling-lemma}
    Let $\nu$ and $\mu$ be two probability distributions on $\Omega$.
    Then
    \begin{align*}
        \DTV{\nu}{\mu} = \inf\set{\Pr{X \neq Y} \mid  (X, Y) \text{ is a coupling of $\nu$ and $\mu$}}.
    \end{align*}
    There exists an optimal coupling of $\nu$ and $\mu$ that achieves the infimum.
\end{lemma}

Coupling can be used to bound both mixing time and relaxation time of a Markov chain.
Let  $\rho: V \to \mathbb{N}_{> 0}$ be the Hamming distance weight in~\Cref{def:CI}.
Let $\Omega \subseteq [q]^V$. 
Let $P$ be the transition matrix of a Markov chain on $\Omega$.
Suppose there is a $\delta \in (0, 1)$ that for each $x, y\in \Omega$, there is a coupling $(X_1, Y_1)$ of $P(x, \cdot)$ and $P(y, \cdot)$ satisfying
\begin{align} \label{eq:coupling-contract}
    \E{H_\rho(X_1, Y_1)} \leq (1 - \delta) H_\rho(x, y),
\end{align}
where $H_\rho$ is the weighted Hamming distance in~\eqref{eq:def-hamming}.
For convenience, let $\rho_{\max} := \max_{v \in V} \rho_v$ and $\rho_{\min} := \min_{v \in V} \rho_v$.
Then it is standard to have the following results.
\begin{lemma} \label{lem:coupling-mix}
    If \eqref{eq:coupling-contract} holds, then $T_{\-{mix}}(P, \epsilon) \leq \frac{1}{\delta} \tp{\log \frac{\rho_{\max}}{\rho_{\min}} + \log \frac{1}{\epsilon}}$.
\end{lemma}
\begin{lemma} \label{lem:coupling-relax}
    If \eqref{eq:coupling-contract} holds and the matrix $P$ is positive semi-definite, then the relaxation time is bounded by
    \begin{align*}
        T_{\-{rel}}(P) \leq 1/\delta.
    \end{align*}
\end{lemma}
We refer the reader to~\cite{bubley1997path} for a proof of \Cref{lem:coupling-mix} and~\cite[Theorem 13.1]{levin2017markov} for a proof of~\Cref{lem:coupling-relax}.
In particular, all the Markov chains considered in this paper are down-up walks (i,e. block dynamics).
And it is well-known that down-up walks are positive semi-definite~\cite[Section 2.4]{alev2020improved}.
Hence, when \eqref{eq:coupling-contract} holds, we can directly use \Cref{lem:coupling-relax} to bound the relaxation time of the down-up walk and block dynamics. Finally, the path coupling is a tool to construct the coupling of the Markov chain.

\begin{lemma}[\text{\cite{bubley1997path}}]
Let $P: [q]^V \times [q]^V \to \mathrm{R}_{\geq 0}$ be a Markov chain.
If for any $x,y \in [q]^V$ such that $x$ and $y$ disagree only at a single vertex $v$, there exists a coupling $(x,y) \to (X_1,Y_1)$ of $P$ such that 
\begin{align*}
  \E{\rho(X_1, Y_1)} \leq (1 - \delta) \rho(x, y) = (1-\delta)\rho(v_0),  
\end{align*}
then  for any $x,y \in [q]^V$, there exists a coupling satisfying~\eqref{eq:coupling-contract}.
\end{lemma}
With path coupling, one only needs to construct coupling for adjacent $x,y$ rather than all $x,y$. The path coupling requires that $P$ is defined over $[q]^V \times [q]^V$ rather than $\Omega \times \Omega$, where $\Omega \subseteq [q]^V$. 
In this paper, we only consider block dynamics. It is easy to extend the block dynamics to $[q]^V \times [q]^V$ using the conditional distribution defined in~\eqref{eq:def-conditional}. 


\subsection{Lov{\'a}sz Local Lemma and Algorithmic Local Lemma}
The Lovasz local lemma~\cite{paul1975problems} is used to prove the existence of a combinatorial object.
In this paper, we will focus on its symmetric and algorithmic version~\cite{moser2010constructive}.

Let $\+E_1, \cdots, \+E_n$ be a set of bad events in some probability space.
We want to show that there is a sample in the probability space that is not included in any bad events.
A \emph{dependency graph} for $\+E_1, \cdots, \+E_n$ is a graph $G = ([n], E)$ such that for each $i \in [n]$, event $\+E_i$ is mutually independent of the events $\set{\+E_j \mid (i,j) \not\in E}$.

\begin{lemma}[Lov{\'a}sz local lemma, \cite{paul1975problems,moser2010constructive}] \label{lem:local-lemma}
    Let $p = \max_i \Pr{\+E_i}$ and $\Delta$ be the degree of dependency graph for $\+E_1, \cdots, \+E_n$.
    If $\e p (\Delta + 1) \leq 1$, then it holds that 
    \begin{align*}
        \Pr{\bigcap_{i=1}^n \ol{\+E_i}} > \tp{1 - \frac{1}{\Delta+1}}^n,
    \end{align*}
    where we use $\ol{\+E_i}$ to denote the negation of $\+E_i$.

    In particular, if all the events $\+E_1, \cdots, \+E_n$ are determined by a set of mutually independent random variables $X_1, \cdots X_m$, then there is a Las Vegas algorithm that runs in time 
 $\frac{np}{1-p}$ in expectation. When it halts, it will output an assignment $\*x = (x_1, \cdots, x_m)$ of $X_i$s such that event $\+E_j$ does not holds under $\*x$ for all $j$.
\end{lemma}

\section{Proof of Relaxation Time Comparison Result}\label{sec:relcom}

In this section, we will prove \Cref{thm:main}.
Let $\mu$ be a Gibbs distribution on graph $G = (V, E)$.
The proof works whenever the graph $G$ has a good partition.
We give it a formal definition as follow.

\begin{definition}[$(\xi,k)$-degree partition]\label{def:partition}
Let $k \geq 1$ be an integer and $\xi > 0$.
A graph $G=(V,E)$ is said to have a $(\xi,k)$-degree partition if there exists a partition   $V = U_1 \uplus U_2 \uplus \ldots \uplus U_k$ such that 
\begin{align*}
\forall 1\leq i \leq k, \forall v \in V, \quad
    |\Gamma_v \cap U_i| \leq \frac{(1+\xi)\Delta}{k}. 
\end{align*}
\end{definition}

We note that a similar but simpler definition appeared in \cite{JainSS21}, where they partition the graph into $2$ parts with certain degree constrains. 

The following result states that we can do comparison between relaxation time of Glauber dynamics given a good partition and $M$-coupling-independence.

\begin{theorem} \label{lem:relax-compare-given-partition}
    Let $\mu$ be a Gibbs distribution on $G = (V, E)$ such that:
    \begin{itemize}
        \item $\mu$ satisfies $M$-coupling-independence for some integer $M \geq 1$;
        \item $G$ has a $(\xi, k)$-degree partition such that $k \geq 2M$.
    \end{itemize}
    Then, the relaxation time of Glauber dynamics on $\mu$ satisfies
    \begin{align*}
     T_{\textnormal{rel}}^{\textnormal{GD}}(\mu) \leq 2^{k-2M}\cdot T_{\textnormal{rel}}^{(\eta)}(\mu), \quad \text{such that } \eta = 2(1+\xi)M/k.
    \end{align*}
\end{theorem}

Note that \Cref{lem:relax-compare-given-partition} only need the existence of such partition.
It does not need the explicit construction of the partition.
When the maximum degree $\Delta$ of $G$ is sufficiently large, the existence of $(\zeta,k)$-partition could be proved by the Lov{\'a}sz local lemma.

\begin{proposition}
\label{lem:partition}
Let $k \geq 1$ be an integer and $\xi \in (0, 1)$. 
For any graph $G=(V,E)$ with maximum degree $\Delta \geq \Delta_0(k,\xi) = \Omega(\frac{k^2}{\xi^2}\log k)$ has a $(\xi,k)$-degree partition.
\end{proposition}
We defer the proof of \Cref{lem:partition} to \Cref{sec:partition}.
Now, we are ready to prove \Cref{thm:main}.
\begin{proof}[Proof of \Cref{thm:main}]
    For simplicity, we just fix $\xi = 1$.
    Note that $M$-coupling-independence directly implies $\ctp{M}$-coupling-independence.
    We take $k = \ctp{4\ctp{M}/\eta}$, which implies $k > 2\ctp{M}$ automatically.
    By \Cref{lem:partition}, when $\Delta \geq \Delta_0(k, 1) = \Omega(k^2 \log k) = \Omega(\frac{M^2}{\eta^2}\log \frac{M}{\eta})$, graph $G$ has a $(1, k)$-partition.
    We make a summary of what we have as follow:
    \begin{itemize}
        \item $\mu$ is $\ctp{M}$-coupling-independence;
        \item $G$ has $(1,k)$-degree partition and $k \geq 2\ctp{M}$.
    \end{itemize}
    Then, by \Cref{lem:relax-compare-given-partition}, we have
    \begin{align*}
        T_{\textnormal{rel}}^{\textnormal{GD}}(\mu) \leq 2^{k-2\ctp{M}}\cdot T_{\textnormal{rel}}^{(4\ctp{M}/\ctp{4\ctp{M}/\eta})}(\mu)
        \leq 2^{O(M/\eta)} \cdot T_{\textnormal{rel}}^{(\eta)}(\mu),
    \end{align*}
    where in the last inequality, we use the fact that $k= O(M/\eta)$ and $4\ctp{M} \leq \ctp{4\ctp{M}/\eta} \cdot \eta$.
\end{proof}

Now, we only left to prove \Cref{lem:relax-compare-given-partition}.
Let $\mu$ be a Gibbs distribution on graph $G=(V,E)$. 
Suppose $G$ has a partition $V = U_1 \uplus U_2\uplus \ldots \uplus U_k$.
We note that for $1 \leq i \leq k$, it is possible that $U_i$ is empty.
Let $1 \leq \ell \leq k$ be an integer.
Consider the following $k\leftrightarrow \ell$ down-up walk defined on the partition. It starts from an arbitrary $X \in \Omega(\mu)$ and it does as follows in each step:
\begin{itemize}
    \item pick a subset $R \subseteq [k]$ with $|R| = \ell$ uniformly at random;
    \item resample $X_{V\setminus U_R} \sim \mu_{V\setminus U_R}(\cdot \mid X_R)$, where $U_R = \cup_{i \in R}U_i$.
\end{itemize}
In the down walk, it picks a subset $R \subseteq [k]$ of size $\ell$ and remove the configuration at $V\setminus U_{R}$. In the up walk, it resamples the configuration on $V\setminus U_{R}$ conditional on the configuration on $U_R$.
Let $T_{\text{rel}}(k, \ell)$ denote the relaxation time of the above $k \leftrightarrow \ell$ down-up walk.
This down-up walk is a \emph{block dynamics} for the Gibbs distribution $\mu$. 
In each step, it picks random block $U_R$ and update the values on other variables conditional on $U_R$.

The following result bound the relaxation time of the $k \leftrightarrow (k-2M)$ down-up walk. 
\begin{lemma}\label{thm:block}
%
Let $M\geq 1$ and $k \geq 2M$ be two integers.
For any Gibbs distribution $\mu$ on graph $G$, if $\mu$ satisfies $M$-coupling independence,
then for any partition $U_1, \cdots U_k$ of $G$, 
the relaxation time of $k\leftrightarrow (k-2M)$ down-up walk satisfies $T_\textnormal{rel}(k,k-2M) \leq 2^{k-2M}$.
\end{lemma}
\Cref{thm:block} is proved in \Cref{sec:step-2}.
It follows from the standard local-to-global paradigm.
However, unlike the usual situation, the ``local'' part here is actually non-local in the sense that it may have exponential sized \emph{influence matrix} or \emph{correlation matrix}.
We surpass this obstacle by working with \emph{coupling independence} instead of standard \emph{spectral independence} condition in~\cite{anari2020spectral}.

Then, according to standard comparison argument between the block dynamics and the Glauber dynamics, we can bound the relaxation time of the Glauber dynamics by the relaxation time of the block dynamics.
Taking into account that the block dynamics actually runs on a $(\xi, k)$-partition, we have the following result which is proved in \Cref{sec:step-3}.
\begin{lemma} \label{lem:compare}
    For any Gibbs distribution $\mu$ on graph $G$, if $G$ has a $(\xi, k)$-partition $U_1, \cdots, U_k$.
    Then for any $0 \leq \ell \leq k-1$, the relaxation time of Glauber dynamics is bounded by 
    \begin{align*}
         T_{\textnormal{rel}}^{\textnormal{GD}}(\mu) &\leq T_\textnormal{rel}(k,\ell) \cdot T_{\textnormal{rel}}^{(\eta(\ell))}(\mu), \quad \text{where } \eta(\ell) = (1 + \xi) \cdot \frac{k-\ell}{k},
    \end{align*}
    and $T_\textnormal{rel}(k,\ell)$ denotes the relaxation time of the $k \leftrightarrow \ell$ down-up walk on the partition $U_1, \cdots, U_k$.
\end{lemma}

Now, we are ready to prove \Cref{lem:relax-compare-given-partition}.
\begin{proof}[Proof of \Cref{lem:relax-compare-given-partition}]
    Pick $\ell = k - 2M$ and combine \Cref{thm:block} and \Cref{lem:compare}.
\end{proof}


\subsection{Relaxation Time of Block Dynamics via Coupling Independence} \label{sec:step-2}
In this section, we prove \Cref{thm:block}.
To do so, we introduce some auxiliary Markov chains to help us analyze the relaxation time of $k \leftrightarrow (k-2M)$ down up walk. Let $R \subseteq [k]$ with size $|R| = r$. For any $\tau \in [q]^{U_{R}}$, we define the $(k-r) \leftrightarrow 1$ down up walk on $\mu^\tau_{V\setminus U_R}$. The chain starts from arbitrary $Y \in \Omega(\mu^\tau_{V\setminus U_R})$. In each step, it does as follows
\begin{itemize}
    \item pick an index $j \in [k]\setminus R$ uniformly at random and let $\Lambda = V\setminus R \setminus U_j$;
    \item resample $Y_{\Lambda} \sim \mu_{\Lambda}(\cdot \mid Y_{U_j}, \tau)$.
\end{itemize}
In words, in this Markov chain, the configuration on $R$ is fixed by $\tau$ and the configuration on $V\setminus U_R$ is free. Hence, we have $k-r$ levels in total. In the down walk, it picks one index $j$ and remove the configurations on all $V\setminus U_R$ except $U_j$, and thus the chain goes to level 1. In the up walk, it goes back to the level $s$ by sampling a random configuration on $\Lambda$ conditional on $Y_{U_j}$ and $\tau$.
Again, let $T^{\tau}_{\text{rel}}(k-r,1)$ denote the relaxation time of $(k-r) \leftrightarrow 1$ down-up walk conditioning on $\tau$. 
\begin{proposition}[\cite{alev2020improved}]\label{prop:local2global}
Let $0 \leq \ell \leq k-1$.
If there exists $\gamma_0, \gamma_{1},\ldots,\gamma_{\ell-1}$ such that for any $0 \leq r \leq \ell-1$, any $R \subseteq [k]$ with $|R| = r$, any $\tau \in [q]^{U_{R}}$, $T^\tau_{\textnormal{rel}}(k-r,1) \leq \gamma_r$, then the $k \leftrightarrow \ell$ down-up walk satisfies
\begin{align*}
 T_{\textnormal{rel}}(k, \ell) \leq  \prod_{i = 0}^{\ell-1} \gamma_i. 
\end{align*}
\end{proposition}
The proof of \Cref{prop:local2global} follows from standard local to global argument in~\cite{alev2020improved}. We give a simple proof in \Cref{app:l2g} for the completeness.

Finally, the relaxation time bound for $(k-r) \leftrightarrow 1$ down-up walk could be established from coupling independence via path coupling.

\begin{lemma}\label{lem:pathcoupling}
If $\mu$ satisfies $M$-coupling-independent, then for any $0 \leq r \leq k - 2M$, any $R \subseteq [k]$ with $|R| = r$, and any pinning $\tau \in [q]^{U_{{R}}}$,  it holds that
\begin{align*}
 T^\tau_{\textnormal{rel}}(k-r,1) \leq \gamma_r = 2.  
\end{align*}
\end{lemma}
\begin{proof}
We use path coupling to analyze the $(k-r) \leftrightarrow 1$ down-up walk. 
For the simplicity of notation, we denote $\ol{U_R} = V\setminus U_R$.
Fix $R \subseteq [k]$ and $\tau \in [q]^{U_{R}}$.
Note that $\mu^\tau_{\ol{U_R}}$ is a distribution over $[q]^{\ol{U_R}}$.
%
One may assume $\ol{U_R}$ is not empty. Otherwise, the lemma is trivial.

Fix two configurations $X$ and $Y$ in $[q]^{\ol{U_R}}$ (not necessarily feasible in $\mu^\tau_{\ol{U_R}}$) that differs only at single vertex $v_0$. We couple the transition $(X,Y) \to (X',Y')$ of $(k-r) \leftrightarrow 1$ down-up walk as follows
\begin{itemize}
    \item sample the same $j \in [k]\setminus R$, construct the same $\Lambda = \ol{U_{R\cup \set{j}}}$;
    \item if $v_0 \notin U_j$, let  $X'_{U_j}= Y'_{U_j} = X_{U_j}=Y_{U_j}$ and  perfectly couple $X_{\Lambda}'$ and $Y_{\Lambda}'$;
    \item if $v_0 \in U_j$, use the coupling in \Cref{def:CI} to sample $(X'',Y'')$ where $X'' \sim \mu^{\tau \land X_{U_j}}$ and $Y'' \sim \mu^{\tau \land Y_{U_j}}$ and let $X' = X''_{\ol{U_R}}$ and $Y' = Y''_{\ol{U_R}}$.
\end{itemize}
We can apply the coupling in \Cref{def:CI} because $\mu$ is $M$-CI. Specifically, we apply the coupling with pinnings $\tau \land X_{U_j}$ and $\tau \land Y_{U_j}$.
Both $X'',Y''$ returned by the coupling are full configurations on $V$. We take partial configurations on $U$ to get $X'$ and $Y'$.
We have $H_\rho(X',Y') = H_\rho(X'',Y'')$ because $X''_{U_R} = Y''_{U_R} = \tau$.
It is straightforward to verify both $X \to X'$ and $Y \to Y'$ follows the transition rule of $(k-r) \leftrightarrow 1$ down-up walk. We have the following result
\begin{align*}
    \E[]{H_\rho(X',Y') \mid X,Y} \leq \frac{M}{k-r} \cdot \rho(v_0) = \frac{M}{k-r} H_\rho(X,Y) \leq \frac{1}{2}H_\rho(X,Y),
\end{align*}
where the last inequality holds because $k-r \geq 2M$.

By path coupling argument, for any $Z_1$ and $Z_2$ in $[q]^U$ ($Z_1$ and $Z_2$ may differ at multi vertices), there exists a coupling $(Z_1,Z_2) \to (Z_1',Z_2')$ such that $\E[]{H_\rho(Z_1',Z_2') \mid Z_1,Z_2} \leq (1 - \frac{1}{2})H_\rho(Z_1,Z_2)$. Hence, the step-wise decay coupling implies the relaxation time $T^\tau_{\text{rel}}(s,1) \leq 2$.  
\end{proof}

Now, we are ready to prove \Cref{thm:block}.
Combining \Cref{prop:local2global} and \Cref{lem:pathcoupling},
\begin{align*}
  T_{\textnormal{rel}}(k,k-2M) \leq  \prod_{i = 0}^{k-2M-1} \gamma_i = 2^{k-2M}.    
\end{align*}

\subsection{Compare Glauber Dynamics to Block Dynamics (Proof of \texorpdfstring{\Cref{lem:compare}}{})} \label{sec:step-3}

In this section, we will prove \Cref{lem:compare}.
Let $0 \leq \ell < k$ be the integer in theorem. 
Recall that $G$ has a $(\xi,k)$-degree partition $V =U_1 \uplus \ldots \uplus U_k$. 
According to \Cref{prop:relax-tensor-contract}, the relaxation time of the $k\to \ell$ down up walk implies the block factorization of variance as follow: for any function $f: \Omega(\mu) \to \^R$, it holds that
\begin{align*}
     \Var[\mu]{f} \leq  \frac{T_\textnormal{rel}(k,\ell)}{\binom{k}{k-\ell}}\sum_{S \subseteq [k]: |S|=k-\ell } \mu[\Var[U_S]{f}] = \frac{T_\textnormal{rel}(k,\ell)}{\binom{k}{k-\ell}}\sum_{S \subseteq [k]: |S|=k-\ell } \sum_{\tau \in [q]^{V \setminus U_S}}\mu_{V \setminus U_S}(\tau) \Var[\mu^\tau]{f}.
\end{align*}

In distribution $\mu^\tau$, consider the graph $G[U_S]$. By \Cref{def:partition}, the maximum degree of the subgraph is at most $(1 + \xi)\frac{k-\ell}{k} \Delta \leq \eta(\ell) \Delta$.
Let $n$ be the number of vertices in $V$.
We can apply \Cref{cond:pin-time} on $\mu^\tau$ to obtain
\begin{align*}
  \Var[\mu^\tau]{f} \leq T_{\textnormal{rel}}^{(\eta(\ell))}(\mu) \frac{1}{n} \sum_{v \in V} \mu^\tau[\Var[v]{f}] =  T_{\textnormal{rel}}^{(\eta(\ell))}(\mu) \frac{1}{n} \sum_{v \in U_S} \mu^\tau[\Var[v]{f}], 
\end{align*}
where the last equation holds because the values on $V \setminus U_i$ are fixed and so the variance is 0. 
Combining the above two inequalities together implies 
\begin{align*}
 \Var[\mu^\tau]{f} &\leq    T_{\textnormal{rel}}^{(\eta(\ell))}(\mu) \frac{T_\textnormal{rel}(k,\ell)}{\binom{k}{k-\ell}}  \cdot \frac{1}{n} \sum_{S \subseteq [k]: |S|=k-\ell } \sum_{v \in U_S} \mu[\Var[v]{f}]\\
 &= T_{\textnormal{rel}}^{(\eta(\ell))}(\mu) \frac{(k-\ell) T_\textnormal{rel}(k,\ell)}{k} \cdot \frac{1}{n} \sum_{v \in V} \mu[\Var[v]{f}] 
 \leq T_{\textnormal{rel}}^{(\eta(\ell))}(\mu) T_\textnormal{rel}(k,\ell) \cdot \frac{1}{n} \sum_{v \in V} \mu[\Var[v]{f}].
\end{align*}
This proves $T_{\textnormal{rel}}^{\textnormal{GD}} \leq T_{\textnormal{rel}}^{(\eta(\ell))}(\mu) \cdot T_\textnormal{rel}(k,\ell)$.

\subsection{Partition the Graph via Local Lemma (Proof of \texorpdfstring{\Cref{lem:partition}}{})} \label{sec:partition}
Let $G = (V,E)$ be a graph. For any vertex $v \in V$, we use $\Gamma_v$ to denote the neighborhood of $v$.




We use the Lov{\'a}sz local lemma to prove the existence.
For each vertex $v$, sample an index $i \in [k]$ uniformly and independently and let $v$ join the set $U_i$. 
 For each $v \in V$, define bad event $B_v$ as there exists $i \in [k]$ such that $|\Gamma_v \cap U_i| > \frac{(1+\xi)\Delta}{k}$. 
Suppose the degree of $v$ is $1 \leq d \leq \Delta$. 
Fix $i \in [k]$.
For each $j \in [d]$, we use $X_j \in \{0,1\}$ to indicate whether the $j$-th neighbor of $v$ belongs to $U_i$. Then
 \begin{align*}
     \Pr{|\Gamma_v \cap U_i| > \frac{(1+\xi)\Delta}{k}} \leq \Pr[]{\sum_{i=1}^d X_j \geq \frac{d}{k} + \frac{\xi \Delta}{k} } \leq \exp\tp{-\frac{2\xi^2\Delta^2}{dk^2}} \leq \exp\tp{-\frac{2\xi^2 \Delta}{k^2}},
 \end{align*}
 where the tail bound follows from Hoeffding's inequality. By a union bound,
 \begin{align*}
     \Pr{B_v} \leq k \exp\tp{-\frac{2\xi^2 \Delta}{k^2}}.
 \end{align*}
 Finally, $B_u$ and $B_v$ are dependent with each other only if $\dist_G(u,v) \leq 2$. The maximum degree of dependency graph is at most $\Delta^2-1$. The Lov{\'a}sz local lemma says the partition in the lemma exists if $e \cdot \Delta^2 \cdot k e^{-{2\xi^2 \Delta}/{k^2}} < 1$, which is true if $\Delta = \Omega(\frac{k^2}{\xi^2}\log k)$.
This finishes the proof of \Cref{lem:partition}.

We remark that to prove the results about Glauber dynamics, we only use this partition in the analysis and we do not need to explicitly construct the partition.


\section{A Fast Sampling Algorithm} \label{sec:algo}
\newcommand{\DU}[1]{P_{#1}}

In this section, we prove \Cref{thm:algo}.
We will use the same setting as \Cref{sec:relcom}.
Assume $G$ has a $(1,k)$-degree partition $V = U_1 \uplus U_2 \uplus \cdots \uplus U_k$.
By \Cref{lem:local-lemma} and \Cref{sec:partition}, not only such partition exists, but we can construct it explicitly in linear time in expectation. 
We can further modify it to obtain a construction algorithm that runs in worst case linear time.
Let $T$ be the running time of the Las Vegas algorithm.
By definition, $T$ is a random variable with $\E{T} = O(n)$.
By the Markov inequality, if we run the Las Vegas algorithm until time $2\E{T}$, then it halts with probability $\Pr{T \leq 2\E{T}} = 1 - \Pr{T \geq 2\E{T}} \geq 1/2$.
Hence, we can run $\log_2 \frac{2}{\epsilon}$ copies of Las Vegas algorithm until time $2\E{T}$.
If there is at least one copy halts, then we use that partition to continue.
If none of them halts, then we output an arbitrary sample.
The latter case will contribute to the TV-distance but it only happens with probability $(1/2)^{\log_2 \frac{2}{\epsilon}} = \epsilon/2$. 
Hence, in the rest part of this section, we can assume the partition $U_1, \cdots, U_k$ is already provided for simplicity.

We choose $k = \ctp{\frac{4\ctp{M}}{\eta}}$ as we did in the proof of \Cref{thm:main} in \Cref{sec:relcom}, where $M, \eta$ are parameters used in \Cref{thm:algo}.
For convenience, for $R \subseteq [k]$, we will use $U_R$ to denote $\cup_{i\in R} U_i$ and $\Lambda_R = V \setminus U_R$.
Then, for any $R \subseteq [k]$ with $X_{U_R} \in \Omega(\mu_{U_R})$, we will consider the $(k-\abs{R}) \leftrightarrow 1$ down-up walk on the conditional distribution $\mu_{\Lambda_R}(\cdot \mid X_{U_R})$.
To simplify the notation, we will use the following notation.
\begin{definition} \label{def:link}
    We use the notion $\mu[X_{U_R}]$ to denote the distribution $\mu_{\Lambda_R}(\cdot \mid X_{U_R}) = \mu_{V\setminus U_R}(\cdot \mid X_{U_R})$.
\end{definition}
Recall that it starts from an arbitrary $Y = Y_{\Lambda_R} \in \Omega(\mu[X_{U_R}])$, in each step,
\begin{itemize}
    \item pick $i \in [k]\setminus R$  uniformly at random;
    \item resample $Y_{\Lambda_{R\cup \set{i}}} \sim \mu[X_{U_R} \uplus Y_{U_i}]$.
\end{itemize}
In the rest of this section, we will denote this down-up walk as $\DU{X(U_R)}$.

Now, in order to have a fast sampling algorithm for the distribution $\mu$, we try to implement $\DU{\emptyset}$ on this partition.
Note that after picking the index $i \in [k]$, the resample phase of $\DU{\emptyset}$ is non-trivial.
Suppose $i$ is fixed, now the problem is reduced to sampling from $\mu[X_{U_i}]$.
We could then further decompose this task by running $\DU{X(U_i)}$.

We use the above scheme recursively until the problem is reduced to sampling from $\mu[X_{U_R}]$ such that $\abs{R} = k - 2M$.
Then, in order to sample from $\mu[X_{U_R}]$, we simply run the Glauber dynamics for plenty of steps.

Formally, we will use the algorithm $\textsf{SimDownUp}(X, {U_R})$ in \Cref{algo:SimDownUp} to simulate the down-up walk $\DU{X(U_R)}$ or the Glauber dynamics (depends on $\abs{R}$) for plenty of steps.
We will determine the parameters $T_0$ and $T_1$ in \Cref{algo:SimDownUp} later.
This should be able to generate a random sample within small TV-distance from $\mu[X_{U_R}]$.

{\RestyleAlgo{ruled}
\LinesNumbered
\begin{algorithm} 
    \textsf{SimDownUp}$(X, {U_R})$ \Begin{
      Let $Y = X$ be the initial state; \\
      \uIf{$\abs{M} \geq k-2M$}{
        update $Y_{\Lambda_R}$ by running Glauber dynamics on $\mu[X_{U_R}]$ for $T_0$ steps;
      } \Else{
        \For{$t = 1, \cdots, T_1$}{
          pick $i \in [k]\setminus R$ uniformly at random; \\
          update $Y_{\Lambda_{R\cup\set{i}}} = \textsf{SimDownUp}(Y, X_{U_R} \uplus Y_{U_i})$;
        }
      }
      \Return $Y_{\Lambda_R}$;
    }
    \caption{\label{algo:SimDownUp}}
\end{algorithm}
}

Now, we are ready to prove \Cref{thm:algo}.

\begin{proof}[Proof of \Cref{thm:algo}]
    In order to prove \Cref{thm:algo}, Let $X \in \Omega(\mu)$ be an arbitrary state. We only need to verify that $\textsf{SimDownUp}(X, \emptyset)$ will eventually returns a sample $Y$ such that $\DTV{Y}{\mu} \leq \epsilon$ after the time claimed in \Cref{thm:algo}. 
    %
    First, we fix $T_0$ and $T_1$ as follow
    \begin{align*}
        T_0 &= \ctp{T_{\textnormal{mix}}^{\textnormal{GD}}(\mu,\eta) \cdot C \cdot \log\tp{\frac{n}{\epsilon}}}, \\
        T_1 &= \ctp{C \cdot \log \tp{\frac{n}{\epsilon}}},
    \end{align*}
    where $C = O(M/\eta)$.
    Without loss of generality, we may assume that $n$ is sufficiently large such that $2 \leq 2C \leq \log n$ to simplify the formulations.
    Since otherwise, $n$ is a constant and we can run a brute force algorithm to generate samples.
    Then, $T_1 \leq 1 + C \log \frac{n}{\epsilon} \leq \tp{\log \frac{n}{\epsilon}}^2$ and the running time claimed in \Cref{thm:algo} follows directly from our setting for $T_0$ and $T_1$.

    Now, we are only left to show that the sample returned by $\textsf{SimDownUp}(\emptyset)$ has the desired accuracy.
    To achieve this, we will show that for every $R \subseteq [k]$ such that $0 \leq \abs{R} \leq k - 2M$ and for all $X \in \Omega(\mu)$, it holds that
    \begin{align} \label{eq:algo-error-bound}
        \DTV{\textsf{SimDownUp}(X, {U_R})}{\mu[X_{U_R}]} 
        \leq \tp{\frac{\epsilon}{n}}^{C/4} \frac{\rho_{\max}}{\rho_{\min}} \sum_{i=0}^{k - 2M - \abs{R}} T_1^i =: F(\abs{R}).
    \end{align}
    Then, by $T_1 \geq 1$, $k = \ctp{4\ctp{M}/\eta}\leq n$, and the fact that $\rho_{\max}/\rho_{\min} = \-{poly}(n)$, we note that 
    \begin{align*}
        F(0) 
        \leq  \tp{\frac{\epsilon}{n}}^{C/4} \-{poly}(n) T_1^{k} 
        \leq  \tp{\frac{\epsilon}{n}}^{C/4} \-{poly}(n) \tp{\log\tp{\frac{n}{\epsilon}}}^{2k},
    \end{align*}
    where in the last inequality, we use the assumption $2 \leq 2C \leq \log n$ make the formulation simple.
    Now, it is direct to see that there exists a $C$ of the order $O(M/\eta)$ such that $F(0) \leq \epsilon$.
    This proves \Cref{thm:algo}.

    Now, we only left to prove \eqref{eq:algo-error-bound}.
    We will prove \eqref{eq:algo-error-bound} via induction on $\abs{R}$.

    \paragraph{Base case $\abs{R} = k - 2M$.}
    In this case, $\textsf{SimDownUp}(X, U_R)$ degenerates to the  Glauber dynamics on $\mu[X_{U_R}]$ that runs $T_0$ steps.
    By our assumption on $T_0$, this gives
    \begin{align*}
        \DTV{\textsf{SimDownUp}(X, {U_R})}{\mu[X_{U_R}]} \leq \tp{\frac{\epsilon}{n}}^C \leq F(k - 2M).
    \end{align*}
    
    \paragraph{Inductive case $0 \leq \abs{R} < k - 2M$.}
    Fix $0 \leq r < k - 2M$. 
    Suppose \eqref{eq:algo-error-bound} holds when $r < \abs{R} \leq k - 2M$, we will show that it also holds for $\abs{R} = r$.
    Recall that we use $\Lambda_R$ to denote $V\setminus U_R$.
    Suppose $X$ is the initial state used in \Cref{algo:SimDownUp}.
    We run $\DU{X({U_R})}$ from $X(\Lambda_R)$ for $T_1$ steps and we denote these steps as $X(\Lambda_R) = Z_0, Z_1, \cdots, Z_{T_1}$.
    Similarly, we denote each steps of $\textsf{SimDownUp}(X, {U_R})$ as $X(\Lambda_R) = Y_0, Y_1, \cdots, Y_{T_1}$.
    We note that $Z_i$s and $Y_i$s here are partial configurations on $\Omega(\mu_{\Lambda_R})$.
    According to this definition we know $Z_{T_1} \sim \DU{X(U_R)}^{T_1}(X(\Lambda_R), \cdot)$ and $Y_{T_1}$ is returned from $\textsf{SimDownUp}(X, {U_R})$.
    By the triangle inequality,
    \begin{align} \label{eq:dtv-triangle}
        \DTV{\textsf{SimDownUp}(X, {U_R})}{\mu[X_{U_R}]} \leq \DTV{Z_{T_1}}{Y_{T_1}} + \DTV{Z_{T_1}}{\mu[X_{U_R}]}.
    \end{align}
    We note that the TV-distance between $Z_{T_1}$ and $\mu[X_{U_R}]$ can be bounded via the decay of the Markovian coupling, which is already done in the proof of \Cref{lem:pathcoupling}.
    This means
    \begin{align} \label{eq:down-up-mixing}
        \DTV{Z_{T_1}}{\mu[X_{U_R}]} \leq \tp{\frac{1}{2}}^{T_1} \cdot \frac{\rho_{\max}}{\rho_{\min}} \cdot n \leq \tp{\frac{\epsilon}{n}}^{C/4} \cdot \frac{\rho_{\max}}{\rho_{\min}},
    \end{align}
    where, in the last equation, we use the fact that $C \geq 4$ to simplify the formula.
    
    In order to bound $\DTV{Z_{T_1}}{Y_{T_1}}$, we use the standard way to construct a coupling between $Z_{T_1}$ and $Y_{T_1}$.
    Note that $\textsf{SimDownUp}(X, {U_R})$ and $\DU{X(U_R)}$ only differs at the resample stage.
    Recall that for $R\subseteq [k]$, we use $\Lambda_R$ to denote $V\setminus U_R$.
    For $t = 1, \cdots T_1$, We build the coupling between $Z_t$ and $Y_t$ recursively as follow:
    \begin{enumerate}
        \item sample the same index $i \in [k]\setminus R$ and let $Z_t(U_i) = Z_{t-1}(U_i)$ and $Y_t(U_i) = Y_{t-1}(U_i)$;
        \item in the resampling stage:
        \begin{itemize}
            \item if $Z_{t-1} = Y_{t-1}$, then $Z_t(\Lambda_{R\cup \set{i}})$ and $Y_t(\Lambda_{R\cup \set{i}})$ are sampled from the optimal coupling between $\mu[X(U_R) \uplus Y_{t-1}(U_i)]$ and $\textsf{SimDownUp}(Y_{t-1}\uplus X(U_R), X(U_R) \uplus Y_{t-1}(U_i))$;
            \item otherwise, if $Z_{t-1} \neq Y_{t-1}$, then $Z_t(\Lambda_{R\cup \set{i}})$ and $Y_t(\Lambda_{R\cup \set{i}})$ are sampled from $\mu[X(U_R) \uplus Y_{t-1}(U_i)]$ and $\textsf{SimDownUp}(Y_{t-1}\uplus X(U_R), X(U_R) \uplus Y_{t-1}(U_i))$ independently.
        \end{itemize}
    \end{enumerate}
    For simplicity, we use $\nu[X(U_R) \uplus Y_{t-1}(U_i)]$ to denote the distribution generated by $\textsf{SimDownUp}(Y_{t-1}\uplus X(U_R), X(U_R) \uplus Y_{t-1}(U_i))$.
    Hence, according to the coupling lemma (\Cref{lem:coupling-lemma}), we have
    \begin{align}
        \nonumber
        \DTV{Z_{T_1}}{Y_{T_1}} &
        \leq \Pr{Z_{T_1} \neq Y_{T_1}} \leq \Pr{\exists t, \text{s.t. $t$ is the first time that } Z_t \neq Y_t} \\
        \nonumber
        (\text{union bound}) \quad &\leq \sum_{t=1}^{T_1} \Pr{Z_t \neq Y_t \text{ and } \forall_{j < t} Z_j = Y_j} 
        \leq \sum_{t=1}^{T_1} \Pr{Z_t \neq Y_t \mid Z_{t-1} = Y_{t-1}} \\
        \nonumber
        &= \sum_{t=1}^{T_1} \E[i, Y_{t-1}]{ \DTV{\mu[X(U_R) \uplus Y_{t-1}(U_i)]}{\nu[X(U_R) \uplus Y_{t-1}(U_i)]} }\\
        \label{eq:dtv-Z-Y}
        (\text{I.H.}) \quad &\leq T_1 \cdot F(\abs{R}+1).
    \end{align}
    Combining \eqref{eq:dtv-triangle}, \eqref{eq:down-up-mixing}, and \eqref{eq:dtv-Z-Y}, we have
    \begin{align*}
        \DTV{\textsf{SimDownUp}(X, {U_R})}{\mu[X(U_R)]} \leq F(\abs{R} + 1) \cdot T_1 + \tp{\frac{\epsilon}{n}}^{C/4} \cdot \frac{\rho_{\max}}{\rho_{\min}}.
    \end{align*}
    Plugging in the definition of $F$ in \eqref{eq:algo-error-bound}, we finish the proof of \eqref{eq:algo-error-bound} and the proof of \Cref{thm:algo}.
\end{proof}

\section{Mixing Time Comparison for Monotone Systems}
\label{sec:mixcom}

In this section, we prove \Cref{thm:mixcom}.
For monotone spin systems, the famous censoring inequality states that censoring updates never decreases the distance to stationary.
We refer the readers to \cite{PW13} and the textbook \cite{levin2017markov} for details.
We refer \Cref{def:monotone-spin-system} for the definition of monotone spin systems.
Then the censoring inequality on monotone spin systems is stated as follow.

\begin{lemma}[\text{\cite{PW13}}] \label{lem:censoring-ineq}
    Let $\mu$ be a monotone spin system.
    Let $v_1, v_2, \cdots, v_m$ be a random sequence of vertices.
    Let $\pi$ be the distribution resulting from updates at $v_1, \cdots, v_m$, starting from the maximum state.
    Let $\nu$ be the distribution resulting from updates at a random subsequence $v_{i_1}, \cdots, v_{i_k}$, also started from the maximum state. The randomness for the updates is independent from the randomness of the sequence $v_1, v_2, \cdots, v_m$ and the subsequence $v_{i_1}, \cdots, v_{i_k}$.
    Then $\pi \preceq \nu$, and
    \begin{align*}
        \DTV{\pi}{\mu} \leq \DTV{\nu}{\mu}.
    \end{align*}
\end{lemma}

For monotone spin systems, the censoring inequality allows us to compare the Glauber dynamics with the simulate algorithm $\textsf{SimDownUp}$ in \Cref{algo:SimDownUp}. 
%
However, to finish the comparison, we need a slightly stronger partition.
We define it formally as follow.
\begin{definition} \label{def:balanced-partition}
    Let $U_1, \cdots, U_k$ be a $(\zeta, k)$-partition of V as defined in \Cref{def:partition}.
    If it further holds that $\abs{U_i} \geq \frac{n}{2k}$, then we say $U_1 \cdots, U_k$ is a \emph{balanced} $(\zeta, k)$-partition.
\end{definition}

The following result shows that such partition exists.

\begin{proposition} \label{prop:balance-partition}
    Let $k \geq 1$ be an integer and $\xi \in (0, 1)$. 
    Then, any $n$-vertices graph $G=(V,E)$ such that $n = \Omega(k \log k)$ and the maximum degree $\Delta = \Omega(\frac{k^2}{\xi^2}\log k)$ has a balanced $(\xi,k)$-degree partition.
\end{proposition}
\begin{proof}
    For each vertex $v$, sample an index $i_v \in [k]$ uniformly at random and let $v$ join the set $U_{i_v}$.
    According to the proof of \Cref{lem:partition} in \Cref{sec:partition} and \Cref{lem:local-lemma}, when $\Delta = \Omega(\frac{k^2}{\zeta^2} \log k)$,
    \begin{align*}
        \Pr{(U_i)_{i\in [k]} \text{ is a $(\zeta, k)$-partition}} \geq \tp{1 - \frac{1}{\Delta + 1}}^n.
    \end{align*}
    Then for any $t \in [k]$, let $Y_v = \*1[i_v = t]$ to indicate that if vertex $v$ is added to $U_t$.
    Let $Y = \sum_v Y_v$ and we have $\E{Y} = n/k$ directly.
    Then by the Chernoff bound, it holds that
    \begin{align*}
        \Pr{\abs{U_t} < \frac{n}{2k}} \leq \Pr{Y \leq \E{Y}/2} \leq \exp\tp{-\E{Y}/8} = \exp\tp{- \frac{n}{8k}}.
    \end{align*}
    Then according to a union bound
    \begin{align*}
        &\hspace{-2cm}\Pr{(U_i)_{i\in [k]} \text{ is a balanced $(\zeta, k)$-partition}} \\
        &\geq 1 - \Pr{(U_i)_{i\in [k]} \text{ is not a $(\zeta, k)$-partition}} - \sum_{t \in [k]} \Pr{\abs{U_t} < \frac{n}{2k}} \\
        &\geq \tp{1 - \frac{1}{\Delta + 1}}^n - k \exp\tp{- \frac{n}{8k}},
    \end{align*}
    which is positive when $n = \Omega(k \log k)$ and $\Delta = \Omega(k)$.
    In summary, when $\Delta = \Omega(\frac{k^2}{\zeta^2} \log k)$ and $n = \Omega(k \log k)$, then such partition exists.
\end{proof}

\begin{proof} [Proof of \Cref{thm:mixcom}]
   
    In this prove, we will basicly use the same setting as the proof of \Cref{thm:algo} in \Cref{sec:algo}.
    However, to apply the censoring inequality, we need to do the following changes:
    \begin{itemize}
        \item In \Cref{sec:algo}, we are using the $(1,k)$-partition for the proof. Here, we use the stronger balanced $(1,k)$-partition. 
        The existence of such partition is guaranteed by \Cref{prop:balance-partition}.
        Since $k = O(M/\eta)$ is only a constant depending on $M$ and $\eta$, we can assume $n = \Omega(k \log k)$ without loss of generality to meet the requirement in \Cref{prop:balance-partition}.
        \item Also note that in the proof of \Cref{thm:algo} in \Cref{sec:algo}, the initial state $X \in \Omega(\mu)$ for the algorithm $\textsf{SimDownUp}(X, \emptyset)$ is picked arbitrarily from $\Omega(\mu)$.
        Here, in order to use censoring inequality, we have to use the maximum state $X_0 \in \Omega(\mu)$ as initial state. 
    \end{itemize}

    Let $T_{\textsf{sim}} = T_{\textnormal{mix}}^{\textnormal{GD}}(\mu,\eta)  \tp{\log \frac{2n}{\epsilon}}^{O(M/\eta)}$ be the number of updates performed by Glauber dynamics required by \Cref{thm:algo}.
    We note that the extra $\Delta$ in \Cref{thm:algo} is used to implement one step of the Glauber dynamics.

    In this manner, the algorithm $\textsf{SimDownUp}(X_0, \emptyset)$ can be seen as repeatedly applying Glauber dynamics updates $P_v$ (see \eqref{eq:GD-decomp}) at some randomly chosen vertex $v$ to the current states initiating from the maximum over the whole state space.
    Suppose the update sequence of this process is 
    \begin{align*}
    \begin{array}{ccc}
    v_{1,1},& \cdots,& v_{1,T_0}, \\
    v_{2, 1},& \cdots,& v_{2, T_0}, \\
    \cdots,& \cdots,& \cdots, \\
    v_{N, 1},& \cdots,& v_{N, T_0},
    \end{array}
    \end{align*}
    where $T_0$ is the parameter defined in \Cref{algo:SimDownUp} and we use $N$ to denote $T_{\textsf{sim}} / T_0$.
    We note that for each $1 \leq i \leq N$, the vertices $\set{v_{i, j}}_{1 \leq j\leq T_0}$ are picked independently from a subset $\Lambda_{R_i} = V \setminus U_{R_i}$ ($R_i \subseteq [k]$) uniformly at random.
    Note that $R_i$ is also picked from $[k]$ according to some distribution.

    We run the Glauber dynamics $P$ from the maximal state for $C T_{\textsf{sim}}$ steps where $C$ is some integer parameter that we will determine later.
    The update sequence of the Glauber dynamics is
    \begin{align*}
    \begin{array}{ccc}
    u_{1,1},& \cdots,& u_{1,C T_0}, \\
    u_{2, 1},& \cdots,& u_{2, C T_0}, \\
    \cdots,& \cdots,& \cdots, \\
    u_{N, T_0},& \cdots,& u_{N, C T_0},
    \end{array}    
    \end{align*}
    where each $u_{i,j}$ is picked from $V$ uniformly at random.
    We consider an indicator $\^I_{i,j} = \*1[u_{i,j} \in \Lambda_{R_i}]$.
    Note that conditioning on $\^I_{i,j} = 1$, it holds that $u_{i,j}$ is distributed uniformly over $\Lambda_{R_i}$.

    Now, we consider a censored Glauber dynamics $P_1$ starting from the maximum state.
    An update $u_{i,j}$ in $P$ is censored $P_1$ if $\^I_{i,j} = 0$.
    After that, we further censor $P_1$ to $P_2$ to make sure that there are at most $T_0$ updates for each row $i$, $1 \leq i\leq N$.

    For fixed $\^I_{\cdot, \cdot}$, 
    let the resulting distribution of $P$ be $\pi_{\^I}$ and the resulting distribution of $P_2$ be $\nu_{\^I}$.
    By the censoring inequality in \Cref{lem:censoring-ineq}
    \begin{align*}
        \DTV{\pi_{\^I}}{\mu} \leq \DTV{\nu_{\^I}}{\mu}.
    \end{align*}
    Let $\rho$ be the resulting distribution of the algorithm $\textsf{SimDownUp}(X_0, \emptyset)$.
    By the triangle inequality:
    \begin{align*}
        \DTV{\pi_{\^I}}{\mu} \leq \DTV{\nu_{\^I}}{\rho} + \DTV{\rho}{\mu}.
    \end{align*}
    Taking expectation w.r.t. $\^I_{\cdot, \cdot}$ at both side, we arrive at
    \begin{align} \label{eq:censor-end}
         \E[\^I]{\DTV{\pi_{\^I}}{\mu}} \leq \E[\^I]{\DTV{\nu_{\^I}}{\rho}} + \DTV{\rho}{\mu}.
    \end{align}
    By a standard coupling argument, it holds that $\DTV{\pi}{\mu} \leq \E[\^I]{\DTV{\pi_{\^I}}{\mu}}$.
    It sufficient for us to bound the RHS of \eqref{eq:censor-end}.
    According to our choice of $T_{\textsf{sim}}$ and \Cref{thm:algo}, we know that $\DTV{\rho}{\mu} \leq \epsilon / 2$.
    Now, it is sufficient for us to show that $\E[\^I]{\DTV{\nu_{\^I}}{\rho}} \leq \epsilon/2$.
    Note that when $\^I_{\cdot,\cdot}$ is picked that $P_2$ has exactly $T_0$ updates in each row, then $P_2$ and $\textsf{SimDownUp}$ becomes the same random process, and hence $\DTV{\nu_{\^I}}{\rho} = 0$.
    Hence we have
    \begin{align}
        \nonumber
        \E[\^I]{\DTV{\nu_{\^I}}{\rho}} 
        &\leq 1 - \Pr[\^I]{P_2 \text{ has exactly $T_0$ updates in each row}} \\
        \nonumber
        &=  \Pr[\^I]{\exists i, P_1 \text{ has $< T_0$ updates in the $i$-th row}}  \\
        \label{eq:censor-bad-event}
       (\text{union bound}) \quad &\leq \sum_i  \Pr[\^I]{P_1 \text{ has $< T_0$ updates in the $i$-th row}}.
    \end{align}
    For the $i$-th row,  according to the definition of the balanced $(1,k)$-partition $(U_i)_{i\in [k]}$ in \Cref{def:balanced-partition}, we note that each $\^I_{i,j} = 1$ with probability $\abs{\Lambda_{R_i}}/\abs{V} \geq \frac{1}{2k}$.
    This means $(\^I_{i,j})_{1 \leq j \leq C T_0}$ are i.i.d random $0/1$ random variables with mean $\E{\^I_{i,j}} \geq \frac{1}{2k}$.
    Let $X = \sum_{1 \leq j\leq C T_0} \^I_{i,j}$, we know $\E{X} \geq \frac{C T_0}{2k}$.
    According to the Chernoff bound
    \begin{align} 
        \nonumber
        \Pr{X < T_0}
        &= \Pr{X < \frac{C T_0}{2k} \cdot \frac{2k}{C}}
        \leq \Pr{X \leq \E{X} \cdot \frac{2k}{C}} \\
        \label{eq:censor-chernoff}
        &\leq \exp\tp{- \frac{\E{X} (1 - 2k/C)^2}{2}} 
        \leq  \exp\tp{- \frac{C T_0}{10k}}
    \end{align}
    where we use the assumption $C \geq 10k$ to simplify the formulation.

    Combining \eqref{eq:censor-bad-event} and \eqref{eq:censor-chernoff}, it holds that
    \begin{align*}
        \E[\^I]{\DTV{\nu_{\^I}}{\rho}} &\leq N \cdot  \exp\tp{- \frac{C T_0}{10k}}.
    \end{align*}
    It is then directly to notice that in order to make $\E[\^I]{\DTV{\nu_{\^I}}{\rho}} \leq \epsilon / 2$, we need
    \begin{align*}
        C \geq \frac{10 k}{T_0} \log \frac{2N}{\epsilon}.
    \end{align*}
    Recall the parameters related to $\textsf{SimDownUp}$ that we have $k = O(M/\eta)$ and 
    \begin{align*}
        T_0 &= T_{\textnormal{mix}}^{\textnormal{GD}} \cdot O(M/\eta) \cdot \log\tp{\frac{2n}{\epsilon}}, \\
        N &\leq T_{\textsf{sim}} = T_{\textnormal{mix}}^{\textnormal{GD}}(\mu,\eta)  \tp{\log \frac{2n}{\epsilon}}^{O(M/\eta)}.
    \end{align*}
    This means we can assume $n$ to be sufficiently large and pick $C = O(M/\eta)$.
    We finish the proof by noticing that the number of updates performed by the Glauber dynamics is bounded by $C T_{\textsf{sim}} = T_{\textnormal{mix}}^{\textnormal{GD}}(\mu,\eta)  \tp{\log \frac{2n}{\epsilon}}^{O(M/\eta)}$.
\end{proof}

\section{Establish Coupling Independence via Self-Avoiding-Walk Tree}\label{sec:CI}
Let $\mu$ be the Gibbs distribution of a $2$-spin system with parameters $\beta,\gamma,\lambda$ on a connected graph $G = (V, E)$.
Assume there is an arbitrary ordering of vertices in $G$.
We fix an arbitrary total order $<$ on $V$.
For any vertex $r \in V$, the self-avoiding-walk (SAW) tree~\cite{weitz2006counting,LLY13} $T_{\text{SAW}}(G,r)$ is a tree with pinnings on some leaves. The tree enumerates all self-avoiding-walks $v_0,v_1,\ldots,v_\ell$ in $G$ starting from $r = v_0$ such that 
\begin{itemize}
    \item all vertices $v_0,v_1,\ldots,v_{\ell-1}$ are distinct;
    \item either the degree of $v_\ell$ is 1 in $G$  or $v_\ell$ is a cycle-closing vertex (i.e. $v_i = v_\ell$ for some $i < \ell$);
    \item for every cycle-closing vertex $v_\ell$ in a SAW $v_0,v_1,\ldots,v_\ell$ with $v_\ell = v_i$ for some $i < \ell$, the value of $v_\ell$ is fixed as $-$ if $v_{i+1} > v_{\ell - 1}$ and  the value of $v_\ell$ is fixed as $+$ if $v_{i+1} < v_{\ell - 1}$.
\end{itemize}
In the definition of the SAW tree, each vertex $v \in V$ in graph $G$ may have multiple copies in $T_{\text{SAW}}(G,r)$. We say a copy is free if its value is not fixed.
One can define a two-spin system with pinning on $T_{\text{SAW}}$ using the same parameters $\beta,\gamma,\lambda$. We use $\pi$ to denote the Gibbs distribution on $T_{\text{SAW}}(G,r)$. 

We can extend a pinning in $G$ to a pinning in $T_{\text{SAW}}(G,r)$. Let $\Lambda \subseteq V$ and $r \notin \Lambda$. Given any pinning $\tau \in \{-,+\}^\Lambda$, for any vertex $v \in \Lambda$, we find all copies $\hat{v}$ of $v$ in $T_{\text{SAW}}(G,r)$ such that $\hat{v}$ is free, fix the value of $\hat{v}$ as $\tau_{v}$, and remove all descendants of $\hat{v}$.
Again, the pinning only appears in leaves of this new tree.
We slightly abuse the notation to denote the new SAW tree by $T_{\text{SAW}}^\tau(G,r)$ and the Gibbs distribution in this SAW tree by $\pi^\tau$. We also call the SAW tree $T_{\text{SAW}}^\tau(G,r)$ as $T_{\text{SAW}}(G,r)$ with pinning $\tau$.

The SAW tree also admits an inductive definition (e.g. see \cite{chen2020rapid}), which will be used in our analysis. Suppose $r$ has $d$ neighbors $u_1<u_2<\ldots<u_d$. We split $r$ into $r_1,r_2,\ldots,r_d$ and connect $r_i$ to $u_i$ to obtain a graph $\ol{G}$. Let $U_i$ denote the pinning that fixes the value of $r_j$ with $j < i$ to be $-$ and the value for $r_j$ with $j > i$ to be $+$. To construct the SAW tree $T_{\text{SAW}}^\tau(G,r)$, one can first construct all $T_i = T_{\text{SAW}}^{\tau \land U_i}(\ol{G}-r_i,u_i)$ and then connect $r$ to the root of each $T_i$.

Recall the definition of influence matrix $\Psi^\tau_\pi$ in \eqref{eq:inf-mat}.
We list some standard properties of SAW tree, which are widely-used for establishing the spectral independence for 2-spin systems~\cite{anari2020spectral,chen2020rapid, chen2021rapid}. 
\begin{proposition}[\cite{weitz2006counting,anari2020spectral}]\label{prop:SAW}
    For any SAW tree $T = T_{\text{SAW}}^\tau(G,r)$ with Gibbs distribution $\pi$, 
    \begin{itemize}
        \item $\pi_r = \mu_r^\tau$, where $\mu$ is the Gibbs distribution on graph $G$;
        \item for any free copy $\hat{u}$ of $u$, the degree of $\hat{u}$ in $T$ is the same as the degree of $u$ in $G$;
        \item for any $u,v$ in $T$, any $w$ in the path between $u$ and $v$, $\Psi^\tau_\pi(u,v) = \Psi^\tau_\pi(u,w)\Psi^\tau_\pi(w,v)$.
    \end{itemize}
\end{proposition}

We give the following result that relates the coupling independence to the influences in SAW tree. 
Many previous works proved the  $\ell_\infty$ spectral independence for the SAW trees, which is a sufficient condition for the spectral independence of the original Gibbs distribution on $G$. 
With this lemma, we can transform them into coupling independence results in a black-box manner.
The proof of the lemma is similar to the recursive coupling introduced in~\cite{GMP05}.
\begin{lemma}\label{lem:CI-tool}
For any $\Lambda \subseteq V$, any $\tau \in \{-,+\}^{V \setminus \Lambda}$, any $v \in \Lambda$, there exists a coupling $(X,Y)$ of $\mu^{\tau \land v^-}_\Lambda$ and $\mu^{\tau \land v^+}_\Lambda$, where $\tau \land v^c$ is the condition $\tau$ together with $v$ taking $c\in \{-,+\}$, such that 
\begin{align}\label{eq:inf-couple}
 \forall u \in \Lambda \setminus \{v\},\quad   \Pr[]{X_u \neq Y_u} \leq \sum_{\hat{u}: \hat{u} \textnormal{ is a copy of }u \textnormal{ in } T_{\textnormal{SAW}}^\tau(G,v) } |\Psi_\pi^{\tau} (v, \hat{u})|,
\end{align}
where $\Psi_\pi^{\tau}$ is the influence matrix for Gibbs distribution $\pi^\tau$ in $T_{\textnormal{SAW}}(G,v)$.
\end{lemma}
\begin{proof}
Fix parameters $\beta,\gamma,\lambda$ of the 2-spin system.
For any Gibbs distribution $\mu$ on a graph $G=(V,E)$, any pinning $\tau \in [q]^{V \setminus \Lambda}$, where $\Lambda \subseteq V$, let $|\Lambda|$ denote the number of free variables for $\mu^\tau$.  For any $k \geq 1$, we use a induction proof to show that any conditional Gibbs distribution $\mu^\tau$ with $k$ free variables satisfies~\eqref{eq:inf-couple} for any free variable $v$, where $\mu$ can be defined on an arbitrary finite graph as long as $\mu^\tau$ has $k$ free variables. 
The base case $k = 1$ is trivial.

Suppose the induction hypothesis holds for all $k' < k$. We prove it for conditional Gibbs distribution $\mu^\tau$ with $k=|\Lambda|$ free variables, where $\mu$ is defined on $G=(V,E)$ and $\Lambda \subseteq V$.
%
Let $u_1<u_2< \cdots< u_d$ be the neighbors of $v$ in $G$.
Then, we construct a graph $\ol{G}$ from $G$ by splitting $v$ into $\{v_1, v_2, \cdots, v_d\}$. 
Using the same parameters $\beta,\gamma,\lambda$, we can define a Gibbs distribution $\ol{\mu}$ on $\overline{G}$.
Define the pinning $\sigma_i$ such that all $v_j$ for $j \leq i$ take value - and all $v_j$ for $j > i$ take the value $+$.
To simplify the notation, we denote $\nu = \mu^\tau_{\Lambda}$ and $\nu_i = \ol{\mu}^{\tau \land \sigma_i}_{\Lambda \setminus v}$. 
It is easy to see $\nu^{v^+}_{\Lambda \setminus v} = \nu_0$ and $\nu^{v^-}_{\Lambda \setminus v} = \nu_d$. 
To couple $\nu^{v^+}$ and $\nu^{v^-}$, we only need to couple $\nu_0$ and $\nu_d$.

The construction of the coupling $(X_0,X_d)$ between $\nu_0$ and $\nu_d$ is achieved by path coupling method~\cite{bubley1997path}. Specifically, we construct couplings $(X_{i-1},X_{i})$ between $\nu_{i-1}$ and $\nu_i$ for all $1 \leq i \leq d$. We first sample $(X_0,X_1)$ from the first coupling, then conditional on the value $X_{i-1}$, we sample $X_i$ from the $i$-th coupling. Finally, we get a coupling $(X_0,X_d)$ between $\nu_0$ and $\nu_d$.

Consider the $i$-th coupling $(X_{i-1},X_{i})$. The only difference between $\nu_{i-1}$ and $\nu_i$ is the pinning at vertex $v_i$. The only neighbor of $v_i$ is $u_i$. 
Each distribution $\nu_i$ has $|\Lambda \setminus v| = k -1$ free variables, so we can use the I.H. on distribution $\nu_i$.
The coupling does as follows.
\begin{itemize}
    \item Couple $X_{i-1}(u_i)$ and $X_i(u_i)$ via the optimal coupling between marginal distributions.
    \item If $X_{i-1}(u_i) = X_i(u_i)$, then $X_{i-1} = X_i$ can be coupled perfectly.
    Because, by conditional independence, given the same value $c$ on $u_i$, the conditional distribution on $\Lambda \setminus v$ induced from $\nu_i$ and $\nu_{i-1}$ are the same, i.e. 
    $\nu_{i-1,\Lambda \setminus v}^{u_i^c} = \nu_{i,\Lambda\setminus v}^{u_i^c}$ for any $c \in \{-,+\}$; If $X_{i-1}(u_i) \neq X_i(u_i)$, by I.H., there is a coupling between $\nu_{i, \Lambda \setminus v }^{u_i^-}$ and $\nu_{i, \Lambda \setminus v }^{u_i^+}$ satisfying~\eqref{eq:inf-couple}, and we use this coupling to couple all variables $\Lambda \setminus v$\footnote{We remark that $u_i \in \Lambda \setminus v$ and $u_i$ must take $+$ (and $-$) in $\nu_i^{u_i^+}$ (and $\nu_i^{u_i^-}$).} in $X_{i-1}$ and $X_{i}$.
\end{itemize}
By the construction of coupling, we have for any vertex $u \in \Lambda$ and $u \neq v$,
\begin{align*}
 \Pr[]{X_{i-1}(u) \neq X_i(u)} = \Pr[]{X_{i-1}(u_i) \neq X_i(u_i)} \cdot \Pr[]{X_{i-1}(u) \neq X_i(u) \mid X_{i-1}(u_i) \neq X_i(u_i)}.    
\end{align*}

The first probability on the RHS is the total variation distance between two marginals.
Consider SAW tree $T_{\text{SAW}}^{\tau}(G,v)$.
The root $v$ has $d$ child $u_1,u_2,\ldots,u_d$. We use $T_i$ to denote the subtree rooted at $u_i$.
By the inductive definition of SAW tree, the tree $T_i + v$ with the value of $v$ being $-$ (resp. $+$) is exactly  $T_{\text{SAW}}(\ol{G}, u_i)$ with pinning $\sigma_i \land \tau$ (resp. $\sigma_{i-1} \land \tau$).
Since SAW tree preserves the marginal distribution at the root $u_i$ (the first property in \Cref{prop:SAW}),
\begin{align*}
    \Pr[]{X_{i-1}(u) \neq X_i(u)} = \DTV{ \nu_{i-1,u_i} }{ \nu_{i,u_i} } = |\Psi_{\pi}^\tau(v,u_i)|. 
\end{align*}

Consider the second step of coupling.
Recall that $\nu_{i} = \ol{\mu}^{\tau \land \sigma_i}_{\Lambda \setminus v}$.
Consider the SAW tree $T_{\text{SAW}}(\ol{G},u_i)$ with pinning $\tau \land \sigma_i$. Let $\pi^{\tau \land \sigma_i}_i$ denote its Gibbs distribution.
Let $\Psi_{\pi_i}^{\tau \land \sigma_i}$ denote the influence matrix for $\pi_i^{\tau \land \sigma_i}$. By I.H, for any $u \in \Lambda \setminus v$,
\begin{align*}
 \Pr[]{X_{i-1}(u) \neq X_i(u) \mid X_{i-1}(u_i) \neq X_i(u_i)}  \leq \sum_{\text{$\hat{u}$: copies of $u$ in $T^{\tau \land \sigma_i}_{\text{SAW}}(\ol{G},u_i)$ }}|\Psi_{\pi_i}^{\tau \land \sigma_i}(u_i,\hat{u})|.
\end{align*}
Finally, we need to relate $T^{\tau \land \sigma_i}_{\text{SAW}}(\ol{G},u_i)$ to $T_{\text{SAW}}^\tau(G,v)$. Recall $T_i$ is the $i$-th subtree of the root $v$ in $T_{\text{SAW}}^\tau(G,v)$. 
Again, by the induction definition of the SAW tree, $T_i + v$ with pinning $+$ on $v$ is exactly $T^{\tau \land \sigma_i}_{\text{SAW}}(\ol{G},u_i)$. By the conditional independence property, given the value of $u_i$, every variable in subtree $T_i$ is independent from all variables outside $T_i$. Hence, for any $\hat{u} \in T_i$, we have
\begin{align*}
   \Psi_{\pi_i}^{\tau \land \sigma_i}(u_i,\hat{u}) = \Psi_{\pi}^{\tau}(u_i,\hat{u}). 
\end{align*}

Combining all of above analysis together and using a union bound for path coupling, we have
\begin{align*}
    \Pr[]{X_0(u) \neq X_d(u)} \leq \sum_{i = 1}^d  \Pr[]{X_{i-1}(u) \neq X_i(u)} &\leq \sum_{i=1}^d |\Psi_{\pi}^\tau(v,u_i)| \sum_{\text{$\hat{u}$: copies of $u$ in $T_i$ }}|\Psi_{\pi}^{\tau}(u_i,\hat{u})|\\
\text{(by last property in \Cref{prop:SAW})}\quad    &=\sum_{\text{$\hat{u}$: copies of $u$ in $T^\tau_{\text{SAW}}(G,v)$ }}|\Psi^\tau_\pi(v,\hat{u})|.
\end{align*}
This finishes the induction step of the proof.
\end{proof}

\section{Hardcore Model in Bipartite Graphs}\label{sec:bihardcore}

We prove \Cref{thm:bipart} and \Cref{thm:bipartmix} in this section.
First, let us recall the parameters.
Let $\delta > 0,\theta > 1$ be two constants.
Let $G=(V_L,V_R,E)$ be a bipartite graph. Let $\Delta_L \geq 3$ and $\Delta_R$ denote the degree in the left and right parts respectively. Assume $\lambda \leq (1-\delta)\Delta_L$ and $\Delta_R \leq \theta \Delta_L$. 
Let $V = V_L \uplus V_R$.
Let $\mu$ over $\{-,+\}^V$ denote the hardcore distribution in graph $G$ with fugacity $\lambda$, where for any $X \sim \mu$ corresponds the the independent set $\{v \in V \mid X_v = +\}$ in graph $G$.

The proof can be outlined as follows.
We use $\mu_L$ to denote the marginal distribution on $L$ projected from $\mu$. Note that $\mu_L$ is not a Gibbs distribution on graph $G$.

\subsection{Coupling Independence for Marginal Distributions}\label{sec:hardcore-CI}
\begin{lemma}\label{lem:hardcore-CI}
The marginal $\mu_L$ satisfies $O(1/\delta)$-coupling-independence with $\rho(v) = 1$ for all $v \in V$.
\end{lemma}
The coupling independence of $\mu_L$ can be obtained from \Cref{lem:CI-tool} and the SAW-tree analysis in~\cite{chen2023uniqueness}. Fix a pinning $\Lambda \subseteq V_L$ and $\tau \in \{0,1\}^{V_L \setminus \Lambda}$. For any vertex $v \in \Lambda$, we need to show a good coupling exists for $\mu_L^{\tau \land v^+}$ and $\mu_{L}^{\tau \land v^-}$. 
Similar to \Cref{sec:CI}, starting from a vertex $v$, we can define the SAW tree $T_{\text{SAW}}(G,v)$ and extend the pinning $\tau$ to the SAW tree to obtain a hardcore distribution $\pi^\tau$ in SAW tree. We use $L_k$ to denote the set of vertices in level $k$, where the root $v$ is in the level $k = 0$. Hence, all copies of vertices in $V_L$ are in level $L_k$ for even $k$. When $\lambda \leq (1 - \delta)\lambda_c(\Delta_L)$, the following result is proved in \cite[see Theorem 2, Lemma 63, Lemma 64 and inequality (84)]{chen2023uniqueness}.
\begin{lemma}[\cite{chen2023uniqueness}]\label{lem:CLY}
For any $k \geq 0$, $\sum_{w \in L_{2k}} |\Psi_{\pi}^\tau(v,w)| \leq \frac{\Delta_L}{\Delta_L-1}(1+\lambda)^{\Delta_L}(1-\frac{\delta}{10})^k$.
\end{lemma}

By~\Cref{lem:CI-tool}, we can easily transform the influence bound in \cite{chen2023uniqueness} into a coupling independence result. We now prove \Cref{lem:hardcore-CI}.
\begin{proof}[Proof of \Cref{lem:hardcore-CI}]
Note that $\tau$ in a pinning on $V_L \setminus \Lambda$. Vertices in $V_R \uplus \Lambda$ is free given $\tau$.
There exists a coupling $(X,Y)$ between $\mu^{\tau \land v^+}_{\Lambda \uplus V_R}$ and $\mu^{\tau \land v^-}_{\Lambda \uplus V_R}$ satisfying the condition in \Cref{lem:CI-tool}. We can project both $X$ into $\Lambda \subseteq V_L$ to obtain a coupling $(X_{\Lambda},Y_{\Lambda})$ between  $\mu_{\Lambda}^{\tau \land v^+}$ and $\mu_{\Lambda}^{\tau \land v^-}$.  
The expected Hamming distance $|X_\Lambda \oplus Y_\Lambda| = |\{u \in \Lambda \mid X_u \neq Y_u\}|$ can be bounded by
\begin{align*}
    \E[]{|X_\Lambda \oplus Y_\Lambda|} \leq 1 + \sum_{k \geq 1}\sum_{ w \in L_{2k}} |\Psi_{\pi}^\tau(v,w)| \leq \frac{\Delta_L}{\Delta_L-1}(1+\lambda)^{\Delta_L} \sum_{k \geq 0} (1-\frac{\delta}{10})^k,
\end{align*}
where the last inequality is from~\Cref{lem:CLY}. Finally, note that $\frac{\Delta_L}{\Delta_L-1}\leq 1.5$ and $(1+\lambda)^{\Delta_L} = (1+O(\frac{1}{\Delta_L}))^{\Delta_L} = O(1)$. The expected Hamming distance is at most $O(1/\delta)$. Since this bound holds for any pinning $\tau$, we proved the $O(1/\delta)$-coupling-independence for the marginal distribution $\mu_L$. 
\end{proof}

\subsection{Graph Partition for Bipartite Graphs}

We partition all the vertices in the left part $V_L$ into $k$ disjoint parts $U_1,U_2,\ldots,U_k$ such that for 
\begin{align}\label{eq:right-cons}
   \forall v \in V_R, \forall i \in [k],\quad  |\Gamma_v \cap U_i| \leq \Delta_L,
\end{align}
where $\Gamma_v \subseteq V_L$ denote neighbors of $v \in V_R$ in graph $G$.
Since the degree of $v \in V_R$ is $\theta \Delta_L$, we roughly need to partition $V_L$ into $k = \Omega(\theta)$ parts, which can be achieved via local lemma.

\begin{proposition}\label{lem:hardcore-par}
For any $k \geq  \lceil 2\theta \rceil$,
if $\Delta_L = \Omega(\theta \log (k \theta))$, then there exists a partition  $V_L = U_1 \uplus U_2 \uplus\ldots \uplus U_k$ satisfying~\eqref{eq:right-cons}.
\end{proposition}
\begin{proof}
We use the Lov{\'a}sz local lemma (\Cref{lem:local-lemma}) to prove the existence.
We first set the parameter $k \geq \lceil 2\theta \rceil$.
For each vertex $v \in V_L$, sample an index $i \in [k]$ uniformly and independently and let $v$ join the set $U_i$. 
 For each $u \in V_R$, define bad event $B_u$ as there exists $i \in [k]$ such that $|\Gamma_u \cap U_i| > \Delta_L$. 
Suppose the degree of $v$ is $1 \leq d \leq \Delta_R \leq \theta \Delta_L$. 
We have $\frac{d}{k} + \frac{\Delta_L}{2} \leq \Delta_L$.
Fix $i \in [k]$.
For each $j \in [d]$, we use $X_j \in \{0,1\}$ to indicate whether the $j$-th neighbor of $v$ belongs to $U_i$. Then
 \begin{align*}
     \Pr{|\Gamma_u \cap U_i| > \Delta_L} \leq \Pr[]{\sum_{i=1}^d X_j \geq \frac{d}{k} + \frac{\Delta_L}{2} } \leq \exp\tp{-\frac{\Delta^2_L}{2d}} \leq \exp\tp{-\frac{\Delta_L}{2\theta}},
 \end{align*}
 where the tail bound follows from Hoeffding's inequality. By a union bound,
 \begin{align*}
     \Pr{B_u} \leq k \exp\tp{-\frac{\Delta_L}{2\theta}}.
 \end{align*}
 Finally, $B_u$ and $B_v$ are dependent with each other only if $\dist_G(u,v) \leq 2$. The maximum degree of dependency graph is at most $\theta^2\Delta^2_L$. The Lov{\'a}sz local lemma says the partition in the lemma exists if $e \cdot \theta^2\Delta^2_L \cdot k\exp\tp{-\frac{\Delta_L}{2\theta}} < 1$, which is true if $\Delta_L = \Omega(\theta \log \theta k)$.
\end{proof}

\subsection{Relaxation Time and Mixing Time Bounds}
Let $M= O(1/\delta)$ be the coupling-independence parameter.
Without loss of generality, we assume $M$ is an integer at least 1. If not, we can around $M$ up to an integer.
We set parameter $k = \max\{ \lceil 2\theta \rceil, 10M \}$.
Let $\Delta_0 = \Delta_0(\theta,\delta) = \Theta(\theta \log (k \theta)) = O(\theta \log \frac{\theta}{\delta}) = O_{\delta,\theta}(1)$ be the threshold for $\Delta_L$ in \Cref{lem:hardcore-par} such that the good $k$-partition exists if $\Delta_L \geq \Delta_0$.

Now, we consider two cases: $\Delta_L < \Delta_0 = O(1)$ and $\Delta_L \geq \Delta_0$. The first one is the easier case. We claim the following result.
\begin{proposition}\label{prop:easycase}
 Let $\delta \in (0,1)$ and $\theta > 1$ be two constants. For any hardcore model on a $\theta$-balanced bipartite graph $G$ with fugacity $\lambda$, if $\Delta_L < \Delta_0(\theta,\delta)$ and $\lambda \leq (1-\delta)\lambda_c(\Delta_L)$, then the relaxation time and mixing time of Glauber dynamics is $O_{\delta,\theta}(n)$ and $O_{\delta,\theta}(n\log n)$, where $n$ is the number of vertices in $G$.    
\end{proposition}

\Cref{prop:easycase} can be proved by applying the technique in~\cite{chen2020optimal} to the marginal distribution $\mu_L$ and then compare some block dynamics on $\mu_L$ with the Glauber dynamics on $\mu$. 
The proof in~\cite{chen2020optimal} works for spin systems with bounded maximum degree.
Although $\mu_L$ is not a spin system, the conditional independence results still hold on the power graph $G^2[V_L]$ and the proof technique can be applied.
We give the proof of \Cref{prop:easycase} in \Cref{app:bd}.

With \Cref{prop:easycase}, we only need to consider the large degree case $\Delta \geq \Delta_0$, where the partition in \Cref{lem:hardcore-par} exists. Let $U_1,U_2,\ldots,U_k$ denote the partition of $V_L$.
Similar to \Cref{sec:step-2}, we study the $k \leftrightarrow (k-\ell)$ down-up walk on the partition. 
The only difference is that we now consider the partition on $V_L$, so the down-up walk is defined on the marginal distribution $\mu_L$. 
Specifically, in each step, given the current configuration $X \in \{-,+\}^{V_L}$, the down-up walk does as follows 
\begin{itemize}
   \item pick a subset $S \subseteq [k]$ with $|S| = \ell$ uniformly at random;
   \item resample $X_{U_S} \sim \mu_{L,U_S}(\cdot\mid X_{V_L \setminus U_S}) = \mu_{U_S}(\cdot\mid X_{V_L \setminus U_S})$, where $U_S = \cup_{i \in S}U_i$.
\end{itemize}
Since $\mu_L$ is the marginal distribution on $V_L$ projected from $\mu$, $\mu_{L,U_S}(\cdot\mid X_{V_L \setminus U_S})$ is the same as the conditional distribution $\mu_{U_S}(\cdot\mid X_{V_L \setminus U_S})$ induced from $\mu$. 

Let $S \subseteq [k]$ with size $|S| = s$. Let $U_{[k] \setminus S} = \uplus_{i\in [k] \setminus S}U_i$. For any $\tau \in [q]^{U_{[k] \setminus S}}$, we can similarly define the $s \leftrightarrow 1$ down-up walk for distribution $\mu_L$ as that in \Cref{sec:step-2}.
The only difference between here and \Cref{sec:step-2} is that here we define down-up walks on $\mu_L$, which is not a Gibbs distribution of a spin system but the down-up walks in \Cref{sec:step-2} are defined for Gibbs distributions. 
However, both the local-to-global argument in \Cref{prop:local2global} and the path argument in \Cref{lem:pathcoupling} work for general distributions. 
Hence, the same proof in \Cref{sec:step-2} implies the following result.
\begin{lemma}\label{thm:blockhardcore}
The relaxation time of $k\leftrightarrow (k-2M)$ down-up walk for $\mu_L$ on partition $U_1,U_2,\ldots,U_k$ satisfies 
$T_\textnormal{rel}(k,k-2M) \leq 2^{k-2M}.$
\end{lemma}

\Cref{thm:blockhardcore} shows that  $k\leftrightarrow (k-2M)$ down-up walk for $\mu_L$ has a constant relaxation $O_{\theta,\delta}(1)$. 
Now, we relate this $k\leftrightarrow (k-2M)$ down-up walk for $\mu_L$ to the following block dynamics $\+B$ on $\mu$. The block dynamics maintains a hardcore configuration $X \in \{-,+\}^V$ on the whole graph $G$. In each step,
\begin{itemize}
    \item  pick a subset $S \subseteq [k]$ with $|S| = 2M$ uniformly at random;
    \item resample $X_{U_S \cup V_R} \sim \mu_{U_S \cup V_R}(\cdot\mid X_{V \setminus U_S})$.
\end{itemize}
The difference between the above block dynamics on $\mu$ and down-up walk on $\mu_L$ is that in the above dynamics, we always resample the configuration on $V_R$ in every step.
We prove the following result by comparing $\+B$ to the $k\leftrightarrow (k-2M)$ down-up walk.
\begin{lemma}\label{thm:blockB}
The relaxation time for block dynamics $\+B$ is at most $2^{k-2M}$.
\end{lemma}
\begin{proof}
We decompose the  $k\leftrightarrow (k-2M)$ down-up walk into two steps. Given the current configuration $X \in \{-,+\}^{V_L}$, the down operator $D$ sample a uniformly random subset $S \subseteq [k]$ with $|S| = 2M$ and map $X$ to $X_{V_L \setminus U_S}$, and the up operator complete the partial configuration $X_{V_L \setminus U_S}$ into a configuration $X \in \{-,+\}^{V_L}$ by sampling from the conditional distribution. 
Recall the relaxation time is equivalent to the contraction of the down-walk (\Cref{prop:relax-tensor-contract}).
\Cref{thm:blockhardcore} implies that for any distribution $\pi_L$ over $\{-,+\}^{V_L}$, it holds that 
\begin{align}\label{eq:contraction}
    D_{\chi^2}(\pi_L D \,\|\, \mu_L D ) \leq (1 - C)D_{\chi^2}(\pi_L \,\|\, \mu_L ),
\end{align}
where $0 < C = 2^{-k+2M} < 1$. Similarly, we can define a down operator $D_{\+B}$ for the block dynamics $\+B$. Given a full configuration $Y \in \{-,+\}^V$, let $P_L$ be the projection operator that maps $Y$ to $Y_{V_L}$. Then, the down operator $D_{\+B} = P_{L} D$, where $D$ is the down operator for the $(k-2M)$ down-up walk. For any distribution $\pi \in \{0,1\}^{V}$, let $\pi_L = \pi P_L$, we have
\begin{align}\label{eq:chi}
   D_{\chi^2}(\pi D_{\+B} \,\|\, \mu D_{\+B} ) &=    D_{\chi^2}(\pi P_L D \,\|\, \mu P_L D )\notag\\
   &=D_{\chi^2}(\pi_L D \,\|\, \mu_L D )\notag\\
\text{(by~\eqref{eq:contraction})}\quad    &\leq (1 - C)D_{\chi^2}(\pi_L \,\|\, \mu_L )\notag\\
   &=(1 - C)  D_{\chi^2}(\pi P_L  \,\|\, \mu P_L )\notag\\
\text{(by data-processing inequality)}\quad    &\leq (1 - C)  D_{\chi^2}(\pi \,\|\, \mu).
\end{align}
The above inequality implies the relaxation time of $\+B$.
\end{proof}

Finally, we can prove \Cref{thm:bipart} by comparing $\+B$ to the Glauber dynamics.
\begin{proof}[Proof of \Cref{thm:bipart}]
Let $C = 2^{k-2M}$. Since $k = O(\theta + \frac{1}{\delta})$, we have $C= 2^{O(\theta+1/\delta)}$ is a constant.
By   \Cref{thm:blockB}, we have the following block factorization of variance.  
\begin{align*}
    \forall f: \Omega(\mu) \to \mathbb{R}, \quad \Var[\mu]{f} &\leq C \binom{k}{2M}^{-1} \sum_{S \in \binom{k}{2M}} \mu[\Var[U_S \cup V_R]{f}]\\
    &=C \binom{k}{2M}^{-1} \sum_{S \in \binom{k}{2M}} \sum_{\tau \in  \{-,+\}^{V_L \setminus U_S }} \mu_{V_L \setminus U_S}(\tau)\Var[\mu^\tau]{f}.
\end{align*}
Given any pinning $\tau \in \{-,+\}^{V_L \setminus U_S }$, $\mu^\tau$ is a hardcore model in $G[U_S \cup V_R]$ with boundary condition $\tau$. By \Cref{lem:hardcore-par}, the maximum degree of $G[U_S \cup V_R]$ is at most $\Delta_L$. Since $\lambda \leq (1-\delta)\lambda_c(\Delta_L)$, we know that the Glauber dynamics in $\mu^\tau$ has the relaxation time $2^{O(1/\delta)} n$~\cite{chen2021rapid}. We remark that we can use~\cite{chen2021rapid} because in the graph $G[U_S \cup V_R]$, the degree of \emph{every} vertex is at most $\Delta_L$. We cannot directly apply the previous result on $G$ because we only have the degree bound for the left part of $G$. As a consequence,
\begin{align*}
   \Var[\mu^\tau]{f} \leq 2^{O(1/\delta)} \sum_{u \in V_R \cup U_S} \mu^\tau[\Var[u]{f}]. 
\end{align*}
Putting everything together, we have
\begin{align*}
\Var[\mu]{f} &\leq   C \cdot 2^{O(1/\delta)}  \binom{k}{2M}^{-1} \sum_{S \in \binom{k}{2M}} \sum_{\tau \in  \{-,+\}^{V_L \setminus U_S }} \mu_{V_L \setminus U_S}(\tau)\sum_{u \in V_R \cup U_S} \mu^\tau[\Var[u]{f}]\\
&= C \cdot 2^{O(1/\delta)}  \binom{k}{2M}^{-1} \sum_{S \in \binom{k}{2M}}\sum_{u \in V_R \cup U_S} \mu[\Var[u]{f}].
\end{align*}
In the above summation, every vertex $u \in V_L$ is counted for $\binom{k-1}{2M-1}$ times and every vertex $v \in V_R$ is counted for $\binom{k}{2M}$. Since the variance is non-negative, we have the following bound
\begin{align*}
 \Var[\mu]{f} &\leq C \cdot 2^{O(1/\delta)}  \sum_{u \in V_L \cup V_R} \mu[\Var[u]{f}].  
\end{align*}
This proves the $C \cdot 2^{O(1/\delta)} \cdot n = O_{\delta,\theta}(n)$ relaxation time for Glauber dynamics.
\end{proof}

In the rest part of this section, we will prove \Cref{thm:bipartmix}.
The proof follows the same high-level plan as the proof of \Cref{thm:mixcom}.
First, we need to strengthen the partition used in \Cref{lem:hardcore-par} to its balanced version to apply the censoring inequality. By combining the proof for \Cref{lem:hardcore-par} and \Cref{prop:balance-partition}, one could prove the following result.
\begin{proposition}
For any $k \geq  \lceil 2\theta \rceil$,
if $\Delta_L = \Omega(\theta \log (k \theta))$ and $\abs{V_L} = \Omega(k \log k)$, then there exists a partition  $V_L = U_1 \uplus U_2 \uplus\ldots \uplus U_k$ satisfying~\eqref{eq:right-cons} and for every $i \in [k]$, $\abs{U_i} \geq \frac{\abs{V_L}}{2k}$.   
\end{proposition}

Technically, we need one more observation for the original proof to make it work for the bipartite hardcore model.
Here, we run the same algorithm $\textsf{SimDownUp}$ on the distribution $\mu$.
Recall $U_1, U_2, \cdots, U_k$ is a partition of $V_L$.
So, the difference here is that $U_1, \cdots U_k$ are just disjoint sets rather than a partition of $V = V_L \uplus V_R$.

Actually, it is straightforward to notice that the same proof for \Cref{thm:mixcom} and \Cref{thm:algo} still works in this scenario provided that the down-up walk or Glauber dynamics on the conditional distribution $\mu[X_{U_R}]$ is fast mixing (recall that $\mu[X_{U_R}]$ is the notion we defined in \Cref{def:link}).

When $\abs{R} = k - 2M$, the distribution $\mu[X_{U_R}]$ is a hardcore model on $G[U_S \cup V_R]$ with boundary condition $X_{U_R}$, where we let $U_S = V_L \setminus U_R$.
By \Cref{lem:hardcore-par}, the maximum degree of $G[U_S \cup V_R]$ is at most $\Delta_L$.
Since $\lambda \leq (1 - \delta)\lambda_c(\Delta_L)$, the Glauber dynamics in $\mu[X_{U_R}]$ is fast mixing.
\begin{lemma}[\cite{Chen0YZ22, ChenE22}] \label{lem:hardcore-goodcase-mixing}
    When $\lambda \leq (1 - \delta)\lambda_c(\Delta)$, the Glauber dynamics for the hardcore model on any $n$-vertices graph with maximum degree $\Delta$ has mixing time $O_\delta(n\log n)$.
\end{lemma}

When $\abs{R} < k - 2M$, the mixing of the $(k - \abs{R}) \leftrightarrow 1$-down-up walk on $\mu[X_{U_R}] = \mu_{U_S \cup V_R}(\cdot \mid X_{U_R})$ is guaranteed by the mixing of the $(k - \abs{R}) \leftrightarrow 1$-down-up walk on $\mu_{U_S}(\cdot \mid X_{U_R})$.
This is already observed in \cite{chen2023uniqueness} and we include this result to prove \Cref{thm:bipartmix}.
Note that the original proof in \cite{chen2023uniqueness} is only for Glauber dynamics.
But is is straight forward to generalize the proof for the down-up walks (i.e., block dynamics).

\begin{lemma}[\text{\cite[Lemma 102]{chen2023uniqueness}}] \label{lem:mixing-block-proj}
    Let $P_1$ be the $(k - \abs{R}) \leftrightarrow 1$-down-up walk on $\mu[X_{U_R}]$ and let $P_2$ be the $(k - \abs{R}) \leftrightarrow 1$-down-up walk on $\mu_{U_S}(\cdot \mid X_{U_R})$, then we have for every $x \in \Omega(\mu[X_{U_R}])$ every $t \geq 1$, 
    \begin{align*}
        \DTV{P_1^t(x, \cdot)}{\mu[X_{U_R}]} \leq \DTV{P_2^t(y, \cdot)}{\mu_{U_S}(\cdot \mid X_{U_R})},
    \end{align*}
    where $y = x_{U_S}$ is the projection of $x$ onto $U_S$.
\end{lemma}

\begin{proof}[Proof of \Cref{thm:bipartmix}]
The proof of \Cref{thm:bipartmix} is finished by go through the same proof as \Cref{thm:algo} and \Cref{thm:mixcom} with the mixing time of Glauber dynamics and down-up walks required by the proof being guaranteed by \Cref{lem:hardcore-goodcase-mixing} and \Cref{lem:mixing-block-proj}.
\end{proof}

 \section{List Coloring}\label{sec:ListColoring}
\Cref{main:coloring} is a consequence of \Cref{thm:main} and \Cref{main:coloring-algo} is a consequence of \Cref{thm:algo}.
The proofs of \Cref{main:coloring} and \Cref{main:coloring-algo} are given in \Cref{sec:maintech}. The only missing part is the following coupling independence result.
\begin{lemma} [\text{\cite{feng2021rapid}}] \label{lem:coloring-1.763-CI}
    Let $\delta \in (0, 0.5)$ be a constant and $G = (V, E)$ be a triangle-free graph with maximum degree $\Delta \geq 3$.
    Let $\+L = (\+L_v)_{v \in V}$ be the color list such that
    \begin{align}\label{eq:cond-color}
        \forall v \in V, \quad \abs{\+L_v} - \-{deg}_G(v) \geq (\alpha^\star + \delta - 1)\Delta,
    \end{align}
    where $\alpha^\star \approx 1.763$ is the unique solution to the equation $\alpha = \exp(1/\alpha)$.
    Let $\mu = \mu_{G,\+L}$ be the uniform distribution over all the $\+L$-list-colorings of $G$.
    Then, it holds that $\mu$ is $O(1/\delta)$-coupling independence.
\end{lemma}

In the above lemma, we assume $\delta < 0.5$. If $\delta \geq 0.5$, then $|\+L_v| - \deg_G(v) > 1.1 \Delta$ and optimal relaxation and mixing time can be obtained from path coupling. 

In \cite{feng2021rapid}, the spectral independence result is proved for list-coloring.
The above lemma cannot be obtained by applying \Cref{lem:CI-tool}, because list-coloring is not two spin systems.  
However, the proof technique in~\cite{feng2021rapid} also yields the coupling independence result. 

Given a list coloring instance $G = (V,E),\+L$, a pinning $\tau \in \otimes_{v \in V \setminus \Lambda}\+L_v$ and a vertex $v \in \Lambda$, the proof in~\cite{feng2021rapid} builds the following coupling between $\mu_\Lambda^{\tau \land v^a}$ and $\mu_\Lambda^{\tau\land v^b}$, where $a,b$ are two different colors in $L_v$ and $\mu^{\tau \land v^a}_\Lambda$ is the marginal distribution on $v$ conditional on $\tau$ and $v$ taking the color $a$. 
By self-reducibility, we can turn the pinning $\tau$ to a new list-coloring instance in graph $G \gets G[\Lambda]$.
For any $u \in \Lambda$,  $\+L_u \gets \+L_u \setminus \{\tau_w \mid w \notin \Lambda \land \{u,w\} \in E \}$.
The proof of \cite{feng2021rapid} constructs the following coupling, which is essentially the recursive coupling in~\cite{GMP05}.  
Let $u_1,u_2,\ldots,u_d$ denote the neighbors of $v$ in $G$.
We split $v$ into $d$ vertices $v_1,v_2,\ldots,v_d$ and connect $v_i$ with $u_i$.
Denote the graph as $G_v$. Let $\pi$ denote the uniform distribution of list-colorings in $G_v$, where each color list of $v_i$ is the same as the list of $v$.
Define the pinning $\sigma_i$ on $\{v_1,v_2,\ldots,v_d\}$ as $\sigma_i(v_j) = b$ for $j \leq i$ and $\sigma_i(v_j) =a $ for $j > i$. 
We only need to couple $\pi^{\sigma_0} = \mu_\Lambda^{\tau \land v^a}$ and $\pi^{\sigma_d} = \mu_\Lambda^{\tau \land v^b}$.
To do this, we couple each adjacent pair $\pi^{\sigma_{i-1}}$ and $\pi^{\sigma_i}$ and merge them via path coupling. To couple $X \sim \pi^{\sigma_{i-1}}$ and $Y \sim \pi^{\sigma_i}$, the analysis in \cite{feng2021rapid} essentially consider the following coupling. Note that $\sigma_i$ and $\sigma_{i-1}$ differ only at $v_i$ and the only neighbor of $v_i$ is $u_i$.
\begin{itemize}
    \item Couple $X(u_i)$ and $Y(u_i)$ via the optimal coupling of marginal  distributions at $u_i$;
    \item If $X(u_i) = Y(u_i)$, then couple all other vertices perfectly;
    \item If $X(u_i) \neq Y(u_i)$, then use the above process  recursively (splitting $u_i$ and applying path coupling) to couple $\pi^{\sigma_{i-1}(-i)\land X(u_i)}$ and  $\pi^{\sigma_i(-i)\land Y(u_i)}$, where $\sigma_i(-i) = \{\sigma_i(v_j) \mid j \neq i \}$ and we can remove $v_i$ because conditional on $u_i$, $v_i$ has no effect on other vertices.
\end{itemize}
The above process gives a coupling $(X,Y)$ between $\mu_\Lambda^{\tau \land v^a}$ and $\mu_\Lambda^{\tau\land v^b}$. One can extend it to a coupling between $\mu^{\tau \land v^a}$ and $\mu^{\tau \land v^b}$ by letting $X_{V \setminus \Lambda} = Y_{V \setminus \Lambda} = \tau$.
The analysis in~\cite{feng2021rapid} bound the discrepancy of the above recursively coupling such that
\begin{align*}
    \E[]{H_\rho(X,Y)} \leq \frac{9}{2\delta} + 1[X_v \neq Y_v] = \frac{9}{2\delta} + 1 = O\left(\frac{1}{\delta}\right),
\end{align*}
where $H_\rho(X,Y)$ is the standard Hamming distance ($\rho(u) = 1$ for all $u \in V$) between $X$ and $Y$.

\ifdoubleblind
\else
\paragraph{Acknowledgment} 
We would like to thank Heng Guo for bringing the paper~\cite{JainSS21} to our attention and for the helpful discussions at the beginning of this project.
We thank Eric Vigoda for suggesting a better statement of technical results.
We also thank Charlie Carlson, Yitong Yin, and Xinyuan Zhang for their helpful discussions.

\fi

\bibliographystyle{alpha}
\bibliography{refs.bib}

\appendix

\section{Local to Global Proof}\label{app:l2g}
\begin{proof}[Proof of \Cref{prop:local2global}]
Let $f: \Omega(\mu) \to \mathbb{R}$ be a function.
Let $\Var[\mu]{f}$ denote the variance of $f$ with respect to $\mu$. For any subset $S \subseteq V$, let $\mu[\Var[S]{f}]$ denote the average of the variance $\Var[\mu^\sigma]{f}$ where $\sigma \sim \mu_{V\setminus S}$. For any set $U$, we use $R \sim \binom{U}{i}$ to denote sample a uniform subset of $R \subseteq U$ with $|R| = i$. 
Note that when $\ell = 0$, then \Cref{prop:local2global} becomes trivial.
Hence, without loss of generality, we may assume $\ell \geq 1$.
By the definition of relaxation time $T_{\text{rel}}^{\emptyset}(k,1)$, 
\begin{align*}
  \Var[\mu]{f} \leq  T_{\text{rel}}^\emptyset(k,1) \E[R \sim \binom{[k]}{1}]{ \mu[\Var[\ol{U_R}]{f}] }\leq \gamma_0 \E[R \sim \binom{[k]}{1}]{ \mu[\Var[\ol{U_R}]{f}] },
\end{align*}
where for convenience, we use $\ol{U_R}$ to denote the set $V\setminus R$.
In the above inequality, one needs to deal with variance $\Var[\mu^\tau]{f}$, where $\tau \sim \mu_{U_R}$. To bound it, we consider the $(k-1) \leftrightarrow 1$ down-up walk on distribution $\pi = \mu^\tau_{U_B}$ to obtain
\begin{align*}
\forall g \in \Omega(\pi), \quad \Var[\pi]{g} \leq T^\tau_{\text{rel}}(k-1,1) \E[R' \sim \binom{[k]\setminus R}{1}]  {\pi[\Var[\ol{U_{R'\cup R}}]{g}]}.  
\end{align*}
Since the configuration on $U_R$ is fixed as $\tau$ in $\mu^\tau$, the above inequality implies
\begin{align*}
\Var[\mu^\tau]{f} \leq T^\tau_{\text{rel}}(k-1,1) \E[R' \sim \binom{[k]\setminus R}{1}]  {\mu^\tau[\Var[\ol{U_{R'\cup R}}]{f}]} \leq \gamma_1 \E[R' \sim \binom{[k]\setminus R}{1}]  {\mu^\tau[\Var[\ol{U_{R'\cup R}}]{f}]}.    
\end{align*}
Combining the above inequalities, we have
\begin{align*}
     \Var[\mu]{f} &\leq   \gamma_k\gamma_{k-1} \E[R \sim \binom{[k]}{1}]{ \E[R' \sim \binom{[k]\setminus R}{1}]{\mu [ \Var[\ol{U_{R\cup R'}}] {f} ] } } =  \gamma_k\gamma_{k-1} \E[R \sim \binom{[k]}{2}]{ \mu [ \Var[\ol{U_{R}}] {f} ] }.
\end{align*}
Using the above argument iteratively, we have
\begin{align*}
    \Var[\mu]{f} &\leq  \gamma_0\gamma_{1} \E[R \sim \binom{[k]}{2}]{ \mu [ \Var[\ol{U_{R}}] {f} ] }\\
    &\leq  \gamma_0\gamma_{1} \gamma_2 \E[R \sim \binom{[k]}{3}]{ \mu [ \Var[\ol{U_{R}}] {f} ] }\\
    &\leq \ldots\\
    &\leq \tp{\prod_{i=0}^{\ell-1} \gamma_i} \E[R \sim \binom{[k]}{\ell}]{ \mu [ \Var[U_{B}] {f} ] }.
\end{align*}
This proves the proposition.
\end{proof}

\section{Bipartite Graph Hardcore: Bounded Degree Case}\label{app:bd}
In this section, we give a brief proof of \Cref{prop:easycase}.
The proof applies techniques in \cite{chen2020optimal} to this problem.
We need to do some modifications because we only have coupling independence on the left part.

The maximum degree $\Delta_R$ in the right part is bounded by $\theta \Delta_L = O_{\theta,\delta}(1)$. If $\lambda < \frac{1}{2\theta \Delta_L}$, then the whole hardcore model satisfies $\lambda < \frac{1}{2\Delta}$, where $\Delta = \max\{\Delta_L,\Delta_R\}$.
The proposition follows from standard path coupling~\cite{bubley1997path}. 

Now, we assume $\frac{1}{2\theta \Delta_L}\leq\lambda \leq (1-\delta)\lambda_c(\Delta_L)$. As a consequence, $\lambda = \Theta_{\theta,\delta}(\frac{1}{\Delta})$. 
For the distribution $\mu$ on the entire graph, for any pinning $\sigma \in \{-,+\}^{\Lambda}$ with $\Lambda \subseteq V_L \cup V_R$ and $v \notin \Lambda$,
\begin{align}\label{eq:defb}
    \min\left(\mu^\sigma_{v}(+),\mu^\sigma_{v}(-)\right) \geq b(\theta,\delta) = \Omega_{\theta,\delta}(1).
\end{align}
This property is called $b$-marginal boundedness of $\mu$ in \cite{chen2020optimal}. 
It is easy to see the $b$-marginal boundedness of $\mu$ implies the $b$-marginal boundedness of $\mu_L$.
Since the coupling independence implies the spectral independence, by \Cref{lem:hardcore-CI}, $\mu_L$ is $O(1/\delta)$-spectrally independent.
Define 
\begin{align}\label{eq:def-alpha}
    \alpha  := \left(\frac{b^2}{100 e(\Delta_L+\Delta_R)}\right)^{\Delta_L} = \Theta_{\delta,\theta}(1). 
\end{align}
By \cite[Lemma 2.5]{chen2020optimal}, the distribution $\mu_L$ satisfies the following factorization of entropy
\begin{align}\label{eq:Ent-ten}
\forall f: \{-,+\}^{V_L} \to \mathbb{R}_{\geq 0},\quad  \Ent[\mu_L]{f} \leq \frac{C(\theta,\delta)}{\binom{n_L}{\ell}}\sum_{S \in \binom{V_L}{\ell}}\mu_L[ \Ent[S]{f} ],
\end{align}
where $n_L = |V_L|$ and $\ell = \lceil \alpha n_L \rceil$. The constant $C$ depends only on the marginal boundedness parameter $b$, spectral independence parameter $O(1/\delta)$, and the parameter $\alpha$. Since all of them depend on $\delta,\theta$, it holds that $C = C(\delta,\theta)$ is a constant. \cite[Lemma 2.5]{chen2020optimal} holds for $n_L \geq n_0(\theta,\delta)$. However, if $n_L \leq n_0(\theta,\delta)$, then $n_L,\Delta_L,\Delta_R,\lambda = \Theta_{\theta,\delta}(1)$ are all constants, we can take $C(\theta,\delta)$ sufficiently large to make the above inequality hold.

The inequality mentioned above is referred to as block factorization of entropy.  In a similar manner to the proof of \Cref{thm:blockB}, we consider two Markov chains. The first one is the $n_L \leftrightarrow n_L - \ell$ down-up walk for $\mu_L$. Given any configuration $X \in \{-,+\}^{V_L}$, it updates $X$ as follows.
\begin{itemize}
    \item Down-Walk $D$: sample $S \in \binom{V_L}{\ell}$ uniformly at random and update $X$ to $X_{V_L \setminus S}$;
    \item Up-Walk $U$: extend $X_{V_L \setminus S}$ to a configuration on $V_L$ by sampling $X_S \sim \mu_{L,S}^{X_{V_L \setminus S}}$.
\end{itemize}
The second Markov chain is the block dynamics $\+B$ for $\mu$. Given any configuration $X \in \{-,+\}^{V_L\cup V_R}$ of the entire bipartite graph $G=(V_L\cup V_R,E)$, it updates $X$ as follows.
\begin{itemize}
    \item Down-Walk $D_{\+B}$: sample $S \in \binom{V_L}{\ell}$ uniformly at random and update $X$ to $X_{V_L \setminus S}$;
    \item Up-Walk $U_{\+B}$: extend $X_{V_L \setminus S}$ to a configuration on $V_L \cup V_R$ by sampling $X_{S \cup V_R} \sim \mu_{S\cup V_R}^{X_{V_L \setminus S}}$.
\end{itemize}
The factorization in~\eqref{eq:Ent-ten} implies the down-walk $D$ contracts the KL-divergence by a factor of $1 - \frac{1}{C(\theta,\delta)}$~\cite[Lemma 2.7]{chen2020optimal}. Using the same analysis as that for~\eqref{eq:chi} but replacing $\chi^2$-divergence with KL-divergence, we have for any distribution $\nu$ over $\{-,+\}^{V_L \cup V_R}$, 
\begin{align*}
    \DKL{\nu D_{\+B} }{\mu D_{\+B} } \leq \left(1 - \frac{1}{C(\theta,\delta)}\right)\DKL{\nu}{\mu}.
\end{align*}
The above contraction of KL-divergence implies the following factorization of entropy
\begin{align*}
\forall f: \Omega(\mu) \to \mathbb{R}_{\geq 0},\quad
  \Ent[\mu]{f} \leq \frac{C(\theta,\delta)}{\binom{n_L}{\ell}}\sum_{S \in \binom{V_L}{\ell}}\mu[ \Ent[S \cup V_R]{f} ]. 
\end{align*}

Consider a subset $S \subseteq V_L$. 
Let $\+C(S \cup V_R)$ denote the set of connected components graph $G[S\cup V_R]$. 
Given any $\tau \in \{-,+\}^{V_L \setminus S}$, in the conditional distribution $\mu^\tau$, all components in $\+C(S \cup V_R)$ are mutually  independent. By~\cite{F01,CMT15} (also see \cite[Lemma 4.1]{chen2020optimal}), 
\begin{align}\label{eq:prodent}
  \mu[ \Ent[S \cup V_R]{f} ] \leq  \sum_{U \in \+C(S \cup V_R)} \mu[ \Ent[U]{f} ]. 
\end{align}

For any $U$, any $\xi \in \{-,+\}^{(V_L \cup V_R) \setminus U}$, by \cite[Lemma 4.2]{chen2020optimal},
\begin{align}\label{eq:baoli}
    \Ent[\mu^\xi]{f} \leq \frac{3|U|^2{\log (1/b)}}{2b^{2|U|+2}}\sum_{v \in U} \mu^\xi[\Ent[v]{f}],
\end{align}
where $b$ is defined in~\eqref{eq:defb}.
Finally, we prove the following lemma for bounding the size of components, which is similar to \cite[Lemma 4.3]{chen2020optimal}.

\begin{lemma}\label{lem:size}
If we sample $S \in \binom{V_L}{\ell}$ uniformly at random, where $\ell = \lceil \alpha n_L \rceil$, then for any vertex $v \in V_L \cup V_R$, 
\begin{align*}
    \Pr[]{|S_v| \geq k} \leq (e(\Delta_L+\Delta_R))^k (2\alpha)^{ (k-1)/\Delta_L},
\end{align*}
where $S_v$ is the component in $G[S \cup V_R]$ containing $v$ and $S_v = \emptyset$ if $v \notin S \cup V_R$.
\end{lemma}

\begin{proof}
Fix a vertex $v \in V_L \cup V_R$. The maximum degree $\Delta$ of the graph is $\max\{\Delta_L,\Delta_R\}$. The number of $U \subseteq V_L\cup V_R$ such $v \in U$, $G[U]$ is connected and $|U|=k$ is at most $(e\Delta)^k$~\cite{BCKL13}.
Note that $G[U]$ has at least $k - 1$ edges because it is connected. Then, $G[U]$ has at least $\ell_k = \lceil(k-1)/\Delta_L\rceil$ vertices from $V_L$, because every edge has one vertex in $V_L$.  By a union bound,
\begin{align*}
    \Pr[]{|S_v| \geq k} \leq (e\Delta)^k \frac{\binom{n_L-\ell_k}{\ell - \ell_k }}{\binom{n_L}{\ell}} \leq (e\Delta)^k \frac{\ell}{n_L} \cdot \frac{\ell - 1}{n_L-1} \cdot \ldots \cdot \frac{\ell - \ell_k + 1}{n_L-\ell_k + 1} \leq (e\Delta)^k (2\alpha)^{\ell_k}.
\end{align*}
The lemma holds because $\Delta \leq \Delta_L+\Delta_R$ and $\alpha < 1/2$.
\end{proof}

Hence, we can bound the entropy of $f$ as follows 
\begin{align*}
    \Ent[\mu]{f} &\leq  \frac{C(\theta,\delta)}{\binom{n_L}{\ell}}\sum_{S \in \binom{V_L}{\ell}}\mu[ \Ent[S \cup V_R]{f} ]\\
(\text{by~\eqref{eq:prodent}})\quad    &\leq \frac{C(\theta,\delta)}{\binom{n_L}{\ell}}\sum_{S \in \binom{V_L}{\ell}} \sum_{U \in C(S \cup V_R)}\mu[ \Ent[U]{f} ]\\
(\text{by~\eqref{eq:baoli}})\quad &\leq \frac{C(\theta,\delta)}{\binom{n_L}{\ell}}\sum_{S \in \binom{V_L}{\ell}} \sum_{U \in C(S \cup V_R)}\frac{3|U|^2{\log (1/b)}}{2b^{2|U|+2}}\sum_{v \in U} \mu[\Ent[v]{f}]\\
&\leq \frac{3C(\theta,\delta)\log(1/b)}{2b^2} \sum_{v \in V_L \cup V_R}\mu[\Ent[v]{f}]\sum_{k \geq 0} \Pr[S \sim \binom{n_L}{\ell}]{|S_v| = k} \frac{3k^2}{b^{2k}}\\
(\text{by \Cref{lem:size}})\quad &\leq \frac{3C(\theta,\delta)\log(1/b)}{2b^2} \sum_{v \in V_L \cup V_R}\mu[\Ent[v]{f}]\sum_{k \geq 0} (e(\Delta_L+\Delta_R))^k (2\alpha)^{ (k-1)/\Delta_L} \frac{3k^2}{b^{2k}}.
\end{align*}
Note that both $C,b,\Delta_L,\Delta_R$ are constants depending on $\delta$ and $\theta$. 
By the definition of $\alpha$ in~\eqref{eq:def-alpha},
\begin{align*}
    (e(\Delta_L + \Delta_R))^{k-1} (2\alpha)^{(k-1)/\Delta_L} = \left(\frac{b^2}{50}\right)^{k-1}.
\end{align*}
We have
\begin{align*}
  \Ent[\mu]{f} &= O_{\theta,\delta}(1) \sum_{v \in V_L \cup V_R} \mu[\Ent[v]{f}]\sum_{k \geq 0}  \left(\frac{1}{50}\right)^{k-1}k^2 =    O_{\theta,\delta}(1) \sum_{v \in V_L \cup V_R} \mu[\Ent[v]{f}].
\end{align*}
The above inequality is the approximate tensorization of entropy for $\mu$, which implies $O_{\theta,\delta}(n)$ relaxation time and $O_{\theta,\delta}(n \log n)$  mixing time for Glauber dynamics~\cite[Fact 3.5]{chen2020optimal}, where $n = |V_L \cup V_R|$. 


%
%
%

\end{document}